\documentclass{article}

\usepackage{algorithmic}
\usepackage{algorithm}

 \usepackage[preprint,nonatbib]{nips_2018}

\usepackage[utf8]{inputenc} 
\usepackage[T1]{fontenc}    
\usepackage{hyperref}       
\usepackage{url}            
\usepackage{booktabs}       
\usepackage{amsfonts}       
\usepackage{nicefrac}       
\usepackage{microtype}      
\usepackage{color}

\usepackage{graphicx}
\usepackage{subfigure}
\usepackage{amsmath, amsfonts, bm, amsthm}
\usepackage[titletoc, title]{appendix}

\usepackage{chngcntr}
\usepackage{wrapfig}

\newcommand\independent{\protect\mathpalette{\protect\independenT}{\perp}}
\def\independenT#1#2{\mathrel{\rlap{$#1#2$}\mkern2mu{#1#2}}}

\newtheorem{theorem}{Theorem}[section]
\newtheorem{lemma}[theorem]{Lemma}

\newtheorem{corollary}[theorem]{Corollary}
\newtheorem{example}[theorem]{Example}
\newtheorem{assumption}[theorem]{Assumption}
\newtheorem{remark}[theorem]{Remark}
\newcommand{\cov}{\mathrm{cov}}

\newcommand{\st}{~\mathrm{s.t.}~}
\newcommand{\diag}{\mathrm{diag}}
\newcommand{\G}{\mathcal{G}}

\newcommand{\N}{\mathcal{N}}
\newcommand{\R}{\mathbb{R}}
\newcommand{\PP}{\mathbb{P}}

\newcommand{\union}{\cup}

\newcommand{\B}[1]{B^{(#1)}}
\newcommand{\Bk}{B^{(k)}}
\newcommand{\Xk}{X^{(k)}}

\DeclareMathOperator{\pa}{Pa}

\theoremstyle{definition}
\newtheorem{definition}[theorem]{Definition}

\title{Direct Estimation of Differences in Causal Graphs}

\author{
  Yuhao Wang \\
  Lab for Information \& Decision Systems\\
 and Institute for Data, Systems and Society\\
  Massachusetts Institute of Technology\\
  Cambridge, MA 02139 \\
  \texttt{yuhaow@mit.edu} \\
  \And
  Chandler Squires \\
  Lab for Information \& Decision Systems\\
 and Institute for Data, Systems and Society\\
  Massachusetts Institute of Technology\\
  Cambridge, MA 02139 \\
  \texttt{csquires@mit.edu} \\
    \And
     Anastasiya Belyaeva \\
  Lab for Information \& Decision Systems\\
 and Institute for Data, Systems and Society\\
  Massachusetts Institute of Technology\\
  Cambridge, MA 02139 \\
  \texttt{belyaeva@mit.edu} \\
  \And
  Caroline Uhler \\
  Lab for Information \& Decision Systems\\
 and Institute for Data, Systems and Society\\
  Massachusetts Institute of Technology\\
  Cambridge, MA 02139 \\
  \texttt{cuhler@mit.edu }
}

\begin{document}

\maketitle

\begin{abstract}
We consider the problem of estimating the differences between two causal directed acyclic graph (DAG) models with a shared topological order given i.i.d.~samples from each model. This is of interest for example in genomics, where changes in the structure or edge weights of the underlying causal graphs reflect alterations in the gene regulatory networks. We here provide the first provably consistent method for directly estimating the differences in a pair of causal DAGs without separately learning two possibly large and dense DAG models and computing their difference. Our two-step algorithm first uses invariance tests between regression coefficients of the two data sets to estimate the skeleton of the difference graph and then orients some of the edges using invariance tests between regression residual variances. We demonstrate the properties of our method through a simulation study and apply it to the analysis of gene expression data from ovarian cancer and during T-cell activation. 
\end{abstract}

\section{Introduction} \label{sec:intro}

Directed acyclic graph (DAG) models, also known as Bayesian networks, are widely used to model causal relationships in complex systems. Learning the \emph{causal} \emph{DAG} from observations on the nodes is an important problem across disciplines~\cite{FLN00,Pearl:00,RHB00,SGS00}. 
A variety of causal inference algorithms based on observational data have been developed, including the prominent PC~\cite{SGS00} and GES~\cite{Meek1997} algorithms, among others~\cite{SWM17,TBA06}. However, 
these methods require strong assumptions~\cite{URB13}; 
in particular, theoretical analysis of the PC~\cite{KB07} and GES~\cite{NHM15,GB13} algorithms have shown that these methods are usually not consistent in the high-dimensional setting, i.e.~when the number of nodes is of the same order or exceeds the number of samples, unless highly restrictive assumptions on the sparsity and/or the maximum degree of the underlying DAG are made.


The presence of high degree hub nodes is a well-known feature of many networks~\cite{BGL11,BZ04}, thereby limiting the direct applicability of causal inference algorithms. 
However, in many applications, the end goal is not to recover the full causal DAG but 
to detect changes in the causal relations between two related networks. For example, in the analysis of EEG signals it is of interest to detect neurons or different brain regions that interact differently when the subject is performing different activities~\cite{SC13}; in biological pathways genes may control different sets of target genes under different cellular contexts or disease states~\cite{HRD09,POK07}.
Due to recent technological advances that allow the collection of large-scale EEG or single-cell gene expression data sets in different contexts there is a growing need for methods that can accurately identify differences in the underlying regulatory networks and thereby provide key insights into the underlying system~\cite{HRD09,POK07}. The limitations of causal inference algorithms for accurately learning large causal networks with hub nodes and the fact that the difference of two related networks is often sparse call for methods that directly learn the difference of two causal networks without having to estimate each network separately.


The complimentary problem to learning the difference of two DAG models is the problem of inferring the causal structure that is \emph{invariant} across different environments. Algorithms for this problem have been developed in recent literature~\cite{GSKZ17,PBM16,ZHZ17}. However, note that the difference DAG can only be inferred from the invariant causal structure if the two DAGs are known. The problem of learning the difference between two networks has been considered previously in the \emph{undirected} setting~\cite{LFS17,LQG14,ZCL14}. 
However, the undirected setting is often insufficient: 
only a causal (i.e., directed) network provides insights into the effects of interventions such as knocking out a particular gene. 
In this paper we provide to our knowledge the first provably consistent method for directly inferring the differences between pairs of causal DAG models 
that does not require learning each model separately. 


The remainder of this paper is structured as follows: In Section~\ref{sec:related}, we set up the notation and review related work.  
In Section~\ref{sec:method}, we present our algorithm for directly estimating the difference of causal relationships and in Section~\ref{sec:theory}, we provide consistency guarantees for our algorithm. In Section~\ref{sec:eval}, we evaluate the performance of our algorithm on both simulated and real datasets including gene expression data from ovarian cancer and T cell activation. 

\section{Preliminaries and Related Work}\label{sec:related}


Let $\G = ([p], A)$ be a DAG with node set $[p] := \{1, \cdots, p\}$ and arrow set $A$. 
Without loss of generality we label the nodes such that 
if $i\to j$ in $\G$ then $i<j$ (also known as \emph{topological ordering}). To each node $i$ we associate a random variable $X_i$ and let $\PP$ be a joint distribution over the random vector $X = (X_1, \cdots, X_p)^T$. In this paper, we consider the setting where the causal DAG model is given by a \emph{linear structural equation model (SEM) with Gaussian noise}, i.e., 
\begin{align*}
X = B^T X + \epsilon
\end{align*}
where $B$ (the {\it autoregressive matrix}) is strictly upper triangular consisting of the edge weights of $\G$, i.e., $B_{ij}\neq 0$ if and only if $i\to j$ in $\G$, and the noise $\epsilon\sim\mathcal{N}(0,\Omega)$ 
with $\Omega := \diag(\sigma_1^2, \cdots, \sigma_p^2)$, i.e., there are no latent confounders. 
Denoting by $\Sigma$ the covariance matrix of $X$ and by $\Theta$ its inverse (i.e., the \emph{precision matrix}), a short computation yields
$\Theta = (I - B) \Omega^{-1} (I - B)^T$,
and hence 
\begin{equation}\label{eq:theta}
\Theta_{ij} = -\sigma_j^{-2} B_{ij} + \sum\limits_{k > j} \sigma_k^{-2} B_{ik} B_{jk}, \;\, \forall i \neq j \quad\; \textrm{and} \quad\;
\Theta_{ii} = \sigma_i^{-2} + \sum\limits_{j > i} \sigma_j^{-2} B_{ij}^2, \;\, \forall i \in [p].
\end{equation}
This shows that the support of $\Theta$ is given by the {\it moral graph} of $\G$, obtained by adding an edge between pairs of nodes that have a common child and removing all edge orientations. By the \emph{causal Markov assumption}, which we assume throughout, the missing edges in the moral graph encode a subset of the conditional independence (CI) relations implied by a DAG model on $\G$; 
 the complete set of CI relations is given by the \emph{d-separations} that hold in $\G$~\cite{LAU96}[Section 3.2.2]; i.e., 
$X_i \independent X_j \mid X_S$ in $\PP$ whenever nodes $i$ and $j$ are d-separated in $\G$ given a set $S\subseteq [p]\setminus\{i,j\}$. The \emph{faithfulness assumption}
is the assertion that all CI relations entailed by $\PP$ are implied by d-separation in $\G$.

A standard approach for causal inference is to first infer CI relations from the observations on the nodes of $\G$ and then to use the CI relations to learn the DAG structure. However, several DAGs can encode the same CI relations and therefore, $\G$ can only be identified up to an equivalence class of DAGs, known as the {\it Markov equivalence class} (MEC) of $\G$, which we denote by $[\G]$. In \cite{VP90}, the author gave a graphical characterization of the members of $[\G]$; namely, two DAGs are in the same MEC if and only if they have the same {\it skeleton} (i.e., underlying undirected graph) and the same {\it v-structures} (i.e., induced subgraphs of the form $i\to j\leftarrow k$). $[\G]$ can be represented combinatorially by a partially directed graph with skeleton $\G$ and an arrow for those edges in $\G$ that have the same orientation in all members of $[\G]$. This is known as the \emph{CP-DAG} (or \emph{essential graph}) of $\G$~\cite{AMP97}.

Various algorithms have been developed for learning the MEC of $\G$ given observational data on the nodes, most notably the prominent GES~\cite{Meek1997} and PC algorithms~\cite{SGS00}. While GES is a score-based approach that greedily optimizes a score such as the BIC (Bayesian Information Criterion) over the space of MECs, the PC algorithm views causal inference as a constraint satisfaction problem with the constraints being the CI relations. In a two-stage approach, the PC algorithm first learns the skeleton of the underlying DAG and then determines the v-structures, both from the inferred CI relations. GES and the PC algorithms are provably consistent, meaning they output the correct MEC given an infinite amount of data, under the faithfulness assumption. Unfortunately, this assumption is very sensitive to hypothesis testing errors for inferring CI relations from data and violations are frequent especially in non-sparse graphs~\cite{URB13}. 
If the noise variables in a linear SEM with additive noise are non-Gaussian, the full causal DAG can be identified (as opposed to just its MEC)~\cite{SHHK06}, for example using the prominent LiNGAM algorithm~\cite{SHHK06}. Non-Gaussianity and sparsity of the underlying graph in the high-dimensional setting are crucial for consistency of LiNGAM. 

In this paper, we develop a two-stage approach, similar to the PC algorithm, for directly learning the difference between two linear SEMs with additive Gaussian noise on the DAGs $\G$ and $\mathcal{H}$. Note that naive algorithms that separately estimate $[\G]$ and $[\mathcal{H}]$ and take their differences can only identify edges that appeared/disappeared and cannot identify changes in edge weights (since the DAGs are not identifiable). Our algorithm overcomes this limitation. 
In addition, we show in Section~\ref{sec:theory} that instead of requiring the restrictive faithfulness assumption on both DAGs $\G$ and $\mathcal{H}$, consistency of our algorithm only requires assumptions on the (usually) smaller and sparser network of differences.

Let  $\PP^{(1)}$ and $\PP^{(2)}$ be a pair of linear SEMs with Gaussian noise defined by $(\B1, \epsilon^{(1)})$ and $(\B2, \epsilon^{(2)})$. Throughout, we make the simplifying assumption that both $\B1$ and $\B2$ are strictly upper triangular, i.e., that the underlying DAGs $\G^{(1)}$ and $\G^{(2)}$ share the same topological order. This assumption is reasonable for example in applications to genomics, since genetic interactions may appear or disappear, change edge weights, but generally do not change directions. For example, in biological pathways an upstream gene  does not generally become a downstream gene in different conditions. Hence $\B1 - \B2$ is also strictly upper triangular and we define the {\it difference-DAG} ({\it D-DAG}) of the two models by $\Delta := ([p], A_\Delta)$ with $i \to j \in A_\Delta$ if and only if $B^{(1)}_{ij} \neq B^{(2)}_{ij}$; i.e., an edge $i \rightarrow j$ in $\Delta$ represents a change in the causal effect of $i$ on $j$, including changes in the presence/absence of an effect as well as changes in edge weight. Given i.i.d.~samples from $\PP^{(1)}$ and $\PP^{(2)}$, our goal is to infer $\Delta$. Just like estimating a single causal DAG model, the D-DAG $\Delta$ is in general not completely identifiable, in which case we wish to identify the skeleton $\bar{\Delta}$ as well as a subset of arrows $\tilde{A}_\Delta$.

A simpler task 
is learning differences of undirected graphical models. Let $\Theta^{(1)}$ and $\Theta^{(2)}$ denote the precision matrices corresponding to $\PP^{(1)}$ and $\PP^{(2)}$. The support of $\Theta^{(k)}$ consists of the edges in the undirected graph (UG) models corresponding to $\PP^{(k)}$. We define the {\it difference-UG} ({\it D-UG}) by $\Delta_\Theta := ([p], E_{\Delta_\Theta})$, with $i - j \in E_{\Delta_\Theta}$ if and only if $\Theta^{(1)}_{ij} \neq \Theta^{(2)}_{ij}$ for $i \neq j$. Two recent methods that directly learn the difference of two UG models are KLIEP~\cite{LQG14} and DPM~\cite{ZCL14}; for a review and comparison of these methods see~\cite{LFS17}. 
These methods can be used as a first step towards estimating the D-DAG $\Delta$: 
under genericity assumptions, the formulae for $\Theta_{ij}$ in~\eqref{eq:theta} imply that if $\B1_{ij} \neq \B2_{ij}$ then $\Theta^{(1)}_{ij}\neq \Theta^{(2)}_{ij}$. Hence, the skeleton of $\Delta$ is a subgraph of $\Delta_\Theta$, i.e., $\bar{\Delta}\subseteq\Delta_\Theta$. In the following section we present our algorithm showing how to obtain $\bar{\Delta}$ and determine some of the edge directions in $\Delta$. We end this section with a piece of notation needed for introducing our algorithm; we define the \textit{set of changed nodes} to be $S_\Theta := \big\{i \mid \exists j \in [p] \;\textrm{such that}\; \Theta_{i,j}^{(1)} \neq \Theta_{i,j}^{(2)} \big\}$.


\section{Difference Causal Inference Algorithm}\label{sec:method}

\begin{algorithm}[!b]
	\caption{Difference Causal Inference (DCI) algorithm}
	\label{alg:main}
	\begin{algorithmic}
		\STATE {\bfseries Input:} Sample data $\hat{X}^{(1)}$, $\hat{X}^{(2)}$.
		\STATE {\bfseries Output:} Estimated skeleton $\bar{\Delta}$ and arrows $\tilde{A}_\Theta$ of the D-DAG $\Delta$.
		\vspace{0.1cm}
		\STATE Estimate the D-UG $\Delta_\Theta$ and $S_\Theta$; use Algorithm~\ref{alg:skel} to estimate $\bar{\Delta}$; use Algorithm~\ref{alg:dir} to estimate $\tilde{A}_\Delta$.
	\end{algorithmic}
\end{algorithm}

In Algorithm~\ref{alg:main} we present our \emph{Difference Causal Inference} (\emph{DCI}) algorithm for directly learning the difference between two linear SEMs with additive Gaussian noise given i.i.d.~samples from each model. Our algorithm consists of a two-step approach similar to the PC algorithm. The first step, described in Algorithm~\ref{alg:skel}, estimates the skeleton of the D-DAG by removing edges one-by-one. Algorithm~\ref{alg:skel} takes $\Delta_\Theta$ and $S_\Theta$ as input. In the high-dimensional setting, KLIEP can be used to estimate $\Delta_\Theta$ and $S_\Theta$.
%
%
For completeness, in the Supplementary Material we also provide a constraint-based method that consistently estimates $\Delta_\Theta$ and $S_\Theta$ in the low-dimensional setting for general additive noise models. Finally, $\Delta_\Theta$ can also simply be chosen to be the complete graph with $S_\Theta=[p]$. These different initiations of Algorithm~\ref{alg:skel} are compared via simulations in Section~\ref{sec:eval}. The second step of DCI, described in Algorithm~\ref{alg:dir}, infers some of the edge directions in the D-DAG.  While the PC algorithm uses CI tests based on the partial correlations for inferring the skeleton and for determining edge orientations, DCI tests the invariance of certain \emph{regression coefficients} across the two data sets in the first step and the invariance of certain \emph{regression residual variances} in the second step. These are similar to the regression invariances used in~\cite{GSKZ17} and are introduced in the following definitions.


\begin{definition}
Given $i, j \in [p]$ and $S \subseteq [p] \setminus \{i,j\}$, let $M := \{i\} \cup S$ and let $\beta_{M}^{(k)}$ be the best linear predictor of $X_j^{(k)}$ given $X_M^{(k)}$, i.e., the minimizer of $\mathbb{E}[(X_j^{(k)}-(\beta_M^{(k)})^T X_M^{(k)})^2]$. We define the \emph{regression coefficient} $\beta_{i,j\mid S}^{(k)}$ to be the entry in $\beta_M^{(k)}$ corresponding to $i$. 
\end{definition}

\begin{definition}
For $j \in [p]$ and $S \subseteq [p] \setminus \{j\}$, we define $(\sigma_{j \mid S}^{(k)})^2$ to be the \emph{variance of the regression residual} when regressing $X_j^{(k)}$ onto the random vector $X_S^{(k)}$.
\end{definition}

Note that in general $\beta_{i,j\mid S}^{(k)}\neq \beta_{j,i\mid S}^{(k)}$. Each entry in $B^{(k)}$ can be interpreted as a regression coefficient, namely $B_{ij}^{(k)} = \beta^{(k)}_{i,j\mid (\pa^{(k)}(j)\setminus\{i\})},$ where $\pa^{(k)}(j)$ denotes the parents of node $j$ in $\mathcal{G}^{(k)}$. Thus, when $B_{ij}^{(1)}=B_{ij}^{(2)}$, 
then there exists a set $S$ such that $\beta_{i,j \mid S}^{(k)}$ stays invariant across $k = \{1,2\}$. This motivates using invariances between regression coefficients to determine the skeleton of the D-DAG. 
For orienting edges, observe that when $\sigma_j^{(k)}$ stays invariant across two conditions, $\sigma_{j \mid S}^{(k)}$ would also stay invariant if $S$ is chosen such that $S = \pa^{(1)}(j) \cup \pa^{(2)}(j)$.
This motivates using invariances of residual variances to discover the parents of node $j$ and assign orientations afterwards.
Similar to~\cite{GSKZ17}
we use hypothesis tests based on the F-test for testing the invariance between regression coefficients and residual variances.
See the Supplementary Material for details regarding the construction of these hypothesis tests, the derivation of their asymptotic distribution, and an example outlining the difference of this approach to~\cite{GSKZ17} for invariant structure learning.

\begin{algorithm}[tb]
   \caption{Estimating skeleton of the D-DAG}
   \label{alg:skel}
\begin{algorithmic}
   \STATE {\bfseries Input:} Sample data $\hat{X}^{(1)}$, $\hat{X}^{(2)}$, estimated D-UG $\Delta_\Theta$, estimated set of changed nodes $S_\Theta$.
   \STATE {\bfseries Output:} Estimated skeleton $\bar{\Delta}$.
    \vspace{0.1cm}
   \STATE Set $\bar{\Delta} := \Delta_\Theta$;
   \FOR{each edge $i - j$ in $\bar{\Delta}$}
   \STATE If $\exists S \subseteq S_\Theta \setminus \{i,j\}$ such that $\beta_{i,j\mid S}^{(k)}$ is invariant across $k = \{1,2\}$, delete $i - j$ in $\bar{\Delta}$ and continue to the next edge. Otherwise, continue.
   \ENDFOR
\end{algorithmic}
\end{algorithm}

\begin{algorithm}[tb]
	\caption{Directing edges in the D-DAG}
	\label{alg:dir}
	\begin{algorithmic}
		\STATE {\bfseries Input:} Sample data $\hat{X}^{(1)}$, $\hat{X}^{(2)}$, estimated set of changed nodes $S_\Theta$, estimated skeleton $\bar{\Delta}$.
		\STATE {\bfseries Output:} Estimated set of arrows $\tilde{A}_\Delta$.
		\vspace{0.1cm}
		\STATE Set $\tilde{A}_\Delta := \emptyset$;
		\FOR{each node $j$ incident to at least one undirected edge in $\bar{\Delta}$}
		\STATE If $\exists S \subseteq S_\Theta \setminus \{j\}$ such that $\sigma_{j\mid S}^{(k)}$ is invariant across $k = \{1,2\}$, add $i \to j$ to $\tilde{A}_\Delta$ for all $i \in S$, and add $j \to i$ to $\tilde{A}_\Delta$ for all $i \not\in S$ and continue to the next node. Otherwise, continue.
		\ENDFOR
		\STATE Orient as many undirected edges as possible via graph traversal using the following rule:
		\STATE \quad\quad Orient $i - j$ into $i \to j$ whenever there is a chain $i \to \ell_1 \to \cdots \to \ell_t \to j$.
	\end{algorithmic}
\end{algorithm}

\begin{example}
	\label{ex:DCI}
We end this section with a 4-node example showing how the DCI algorithm works. Let $B^{(1)}$ and $B^{(2)}$ be the autoregressive matrices defined by the edge weights of $\G^{(1)}$ and $\G^{(2)}$ and let the noise variances satisfy the following invariances:
\vspace{-0.3cm}
\begin{center}
\begin{tabular}{c c c c}
\includegraphics[scale=0.8]{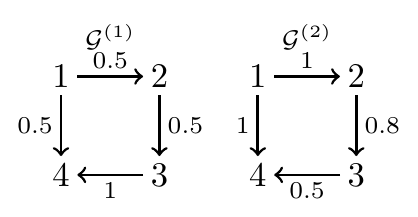}& \includegraphics[scale=1]{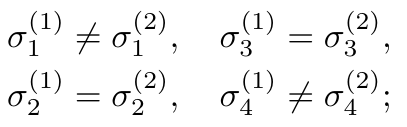} & \includegraphics[scale=1]{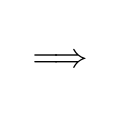}&
	 \includegraphics[scale=0.85]{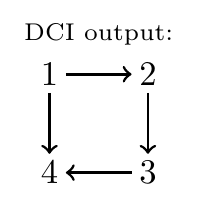}
\end{tabular}
\end{center}
\vspace{-0.3cm}
Initiating Algorithm~\ref{alg:skel} with $\Delta_{\Theta}$ being the complete graph and $S_\Theta=[4]$, the output of the DCI algorithm is shown above. 
\end{example}

\section{Consistency of DCI}\label{sec:theory}

 The DCI algorithm is \emph{consistent} if it outputs a partially oriented graph $\hat\Delta$ that has the same skeleton as the true D-DAG and whose oriented edges are all correctly oriented. Just as methods for estimating a single DAG require assumptions on the underlying model (e.g. the faithfulness assumption) to ensure consistency, our method for estimating the D-DAG requires assumptions on relationships between the two underlying models. To define these assumptions it is helpful to view $(\sigma_j^{(k)})_{j\in[p]}$ and the non-zero entries $(B_{ij}^{(k)})_{(i,j)\in A^{(k)}}$ as \emph{variables} or {\it indeterminates} and each entry of $\Theta^{(k)}$ as a {\it rational function}, i.e., a fraction of two {\it polynomials} in the variables $B_{ij}^{(k)}$ and $\sigma_j^{(k)}$ as  defined in~\eqref{eq:theta}. Using Schur complements one can then similarly express $\beta_{v,w\mid S}^{(k)}$ and $(\sigma_{w \mid S}^{(k)})^2$ as a rational function in the entries of $\Theta^{(k)}$ and hence as a rational function in the variables $(B_{ij}^{(k)})_{(i,j)\in A^{(k)}}$ and $(\sigma_j^{(k)})_{j\in[p]}$. The exact formulae are given in the Supplementary Material.  

 Clearly, if 
 $B_{ij}^{(1)} = B_{ij}^{(2)}$ $\forall (i,j)\textrm{ and }  \sigma_j^{(1)} = \sigma_j^{(2)}$ $\forall j\in[p],$
 then $\beta_{v,w\mid S}^{(1)} = \beta_{v,w\mid S}^{(2)}$ and $\sigma_{w \mid S}^{(1)} = \sigma_{w \mid S}^{(2)}$ for all $v,w,S$. For consistency of the DCI algorithm we assume that the converse is true as well, namely that \emph{differences} in $B_{ij}$ and $\sigma_j$ in the two distributions are not ``cancelled out'' by changes in other variables and result in \emph{differences} in the regression coefficients and regression residual variances. This allows us to deduce invariance patterns of the autoregressive matrix $\Bk$ from invariance patterns of the regression coefficients and residual variances, and hence differences of the two causal DAGs.\footnote{This is similar to the faithfulness assumption in the Gaussian setting, where partial correlations are used for CI testing; the partial correlations are rational functions in the variables  $B_{ij}^{(k)}$ and $\sigma_j^{(k)}$ and the faithfulness assumption asserts that if a partial correlation $\rho_{ij\mid S}$ is zero then the corresponding rational function is identically equal to zero and hence $B_{ij}=0$~\cite{LUS14}.} 

\begin{assumption}
	\label{ass:adj_faith}
For any choice of $\,i, j \in S_\Theta$, if $\B1_{ij} \neq \B2_{ij}$ then for all $S \subseteq S_\Theta \setminus \{i,j \}$~it~holds~that
\vspace{-0.3cm}
\begin{align*}
\beta_{i,j\mid S}^{(1)} \neq \beta_{i,j\mid S}^{(2)} \;\textrm{and}\;  \beta_{j,i\mid S}^{(1)} \neq \beta_{j,i\mid S}^{(2)}.
\end{align*}
\end{assumption}
\vspace{0.1cm}

\begin{assumption}
	\label{ass:ori_faith}
For any choice of $i, j \in S_\Theta$ it holds that
\begin{enumerate}
	\vspace{-0.2cm}
\item if $\B1_{ij} \neq \B2_{ij}$, then $\forall S \subseteq S_\Theta \setminus \{i,j \}$,  $\sigma_{j \mid S}^{(1)} \neq \sigma_{j \mid S}^{(2)} \quad\textrm{and}\quad  \sigma_{i \mid S \cup \{j\}}^{(1)} \neq \sigma_{i \mid S \cup \{j\}}^{(2)}$.
\vspace{-0.2cm}
\item if $\sigma_j^{(1)} \neq \sigma_j^{(2)}$, then $\sigma_{j \mid S}^{(1)} \neq \sigma_{j \mid S}^{(2)}$ for all $S \subseteq S_\Theta \setminus \{j \}$.
\end{enumerate}
\end{assumption}

Assumption~\ref{ass:adj_faith} ensures that the skeleton of the D-DAG is inferred correctly, whereas Assumption~\ref{ass:ori_faith} ensures that the arrows returned by the DCI algorithm are oriented correctly. These assumptions are the equivalent of the \emph{adjacency-faithfulness} and \emph{orientation-faithfulness} assumptions that ensure consistency of the PC algorithm for estimating the MEC of a causal DAG~\cite{RSZ06}. 

We now provide our main results, namely consistency of the DCI algorithm. For simplicity we here discuss the consistency guarantees when Algorithm~\ref{alg:skel} is initialized with $\Delta_\Theta$ being the complete graph and $S_\Theta = [p]$. However, in practice we recommend initialization using KLIEP (see also Section 5) to avoid performing an unnecessarily large number of conditional independence tests. The consistency guarantees for such an initialization including a method for learning the D-DAG in general additive noise models (that are not necessarily Gaussian) is provided in the Supplementary Material.


\begin{theorem}\label{thm:skel}
Given Assumption~\ref{ass:adj_faith}, Algorithm~\ref{alg:skel} is consistent in estimating the skeleton of the D-DAG $\Delta$. 
\end{theorem}
The proof is given in the Supplementary Material. The main ingredient is showing that if $B_{ij}^{(1)} = B_{ij}^{(2)}$, then there exists a conditioning set $S \subseteq S_\Theta \setminus \{i,j\}$ such that $\beta_{i,j \mid S}^{(1)}=\beta_{i,j \mid S}^{(2)}$, namely 
the parents of node $j$ in both DAGs excluding node $i$. 
Next, we provide consistency guarantees for Algorithm~\ref{alg:dir}.

\begin{theorem}\label{thm:dir}
Given Assumption~\ref{ass:ori_faith}, all arrows $\tilde{A}_{\Delta}$ output by Algorithm~\ref{alg:dir} are correctly oriented. In particular, if $\sigma_i^{(k)}$ is invariant across $k = \{1,2\}$, then all edges adjacent to $i$ are oriented.
\end{theorem}

Similar to the proof of Theorem~\ref{thm:skel}, the proof follows by interpreting the rational functions corresponding to regression residual variances in terms of d-connecting paths in $\mathcal{G}^{(k)}$ and is given in the Supplementary Material. It is important to note that as a direct corollary to Theorem~\ref{thm:dir} we obtain sufficient conditions for full identifiability of the D-DAG (i.e., all arrows) using the DCI algorithm.

\begin{corollary}\label{cor:dir}
	Given Assumptions~\ref{ass:adj_faith} and~\ref{ass:ori_faith}, and assuming that the error variances are the same across the two distributions, i.e.~$\Omega^{(1)}=\Omega^{(2)}$, the DCI algorithm outputs the D-DAG $\Delta$.
\end{corollary}


In addition, we conjecture that Algorithm~\ref{alg:dir} is \emph{complete}, i.e., that it directs all edges that are identifiable in the D-DAG. We end this section with two remarks, namely regarding the sample complexity of the DCI algorithm and an evaluation of how restrictive Assumptions~\ref{ass:adj_faith} and~\ref{ass:ori_faith} are.

\vspace{0.1cm}

\begin{remark} [Sample complexity of DCI] \label{rk:complexity} 
For constraint-based methods such as the PC or DCI algorithms, the sample complexity is determined by the number of hypothesis tests performed by the algorithm~\cite{KB07}. In the high-dimensional setting, the number of hypothesis tests performed by the PC algorithm scales as $\mathcal{O}(p^s)$, where $p$ is the number of nodes and $s$ is the maximum degree of the DAG, thereby implying severe restrictions on the {\it sparsity} of the DAG given a reasonable sample size. Meanwhile, the number of hypothesis tests performed by the DCI algorithm scales as $\mathcal{O}(|\Delta_\Theta| 2^{| S_\Theta | - 1})$ and hence does not depend on the degree of the two DAGs. Therefore, even if the two DAGs $\G^{(1)}$ and $\G^{(2)}$ are high-dimensional and highly connected, the DCI algorithm is consistent and has a better sample complexity (as compared to estimating two DAGs separately) as long as the differences between $\G^{(1)}$ and $\G^{(2)}$ are {\it sparse}, i.e., $| S_\Theta |$ is small compared to $p$ and $s$. \qed
\end{remark}

\vspace{0.1cm}


\begin{remark}[Strength of Assumptions~\ref{ass:adj_faith} and~\ref{ass:ori_faith}] \label{rk:assumption}
Since faithfulness, a standard assumption for consistency of causal inference algorithms to estimate an MEC, is known to be restrictive~\cite{URB13}, it is of interest to compare Assumptions~\ref{ass:adj_faith} and~\ref{ass:ori_faith} to the faithfulness assumption of $\,\PP^{(k)}$ with respect to $\G^{(k)}$ for $k\in\{1,2\}$. In the Supplementary Material we provide examples showing that Assumptions~\ref{ass:adj_faith} and~\ref{ass:ori_faith} do not imply the faithfulness assumption on the two distributions and vice-versa. However, in the finite sample regime 
we conjecture 
Assumptions~\ref{ass:adj_faith} and~\ref{ass:ori_faith} to be weaker than the faithfulness assumption: violations of faithfulness as well as of Assumptions~\ref{ass:adj_faith} and~\ref{ass:ori_faith} correspond to points that are close to conditional independence hypersurfaces~\cite{URB13}. The number of these hypersurfaces (and hence the number of violations) increases in $s$ for the faithfulness assumption and in $S_{\Theta}$ for Assumptions~\ref{ass:adj_faith} and~\ref{ass:ori_faith}. Hence if the two DAGs $\G^{(1)}$ and $\G^{(2)}$ are large and complex while having a sparse difference, then $S_{\Theta}<\!\!<s$. See the Supplementary Material for more details. 
\qed
\end{remark}

\section{Evaluation}\label{sec:eval}

In this section, we compare our DCI algorithm with PC and GES on both synthetic and real data. The code utilized for the following experiments can be found at~\url{https://github.com/csquires/dci}.

\subsection{Synthetic data}



We analyze the performance of our algorithm in both, the low- and high-dimensional setting. For both settings we generated 100 realizations of pairs of upper-triangular SEMs $(\B1, \epsilon^{(1)})$ and $(\B2, \epsilon^{(2)})$. For $\B1$, the graphical structure was generated using an Erd\"os-Renyi model with expected neighbourhood size $s$, on $p$ nodes and $n$ samples. The edge weights were uniformly drawn from $[-1, -0.25] \cup [0.25, 1]$ to ensure that they were bounded away from zero. $\B2$ was then generated from $\B1$ by adding and removing edges with probability $0.1$, i.e.,
\begin{align*}
B^{(2)}_{ij} \overset{\textrm{i.i.d.}}\sim\textrm{Ber}(0.9) \cdot B^{(1)}_{ij} \textrm{ if } B^{(1)}_{ij} \!\neq\! 0,
&&
B^{(2)}_{ij} \overset{\textrm{i.i.d.}}\sim\textrm{Ber}(0.1) \cdot \textrm{Unif}([-1, -.25] \cup [.25, 1])\textrm{ if } B^{(1)}_{ij} \!=\! 0
\end{align*}
Note that while the DCI algorithm is able to identify changes in edge weights, we only generated DAG models that differ by edge insertions and deletions. This is to provide a fair comparison to the naive approach, where we separately estimate the two DAGs $\G^{(1)}$ and $\G^{(2)}$ and then take their difference, since this approach can only identify insertions and deletions of edges.

\begin{figure}[t!]
	\centering
	\subfigure[skeleton]{\includegraphics[width=.4\textwidth]{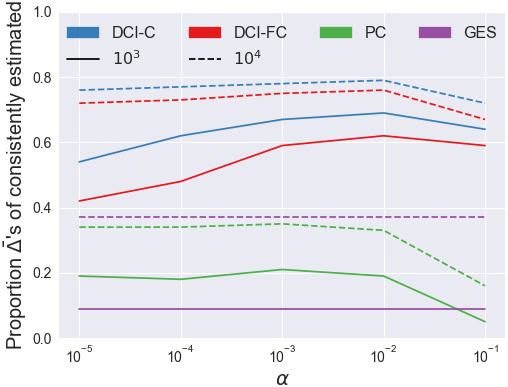}}
	\subfigure[skeleton \& orientation]{\includegraphics[width=.4\textwidth]{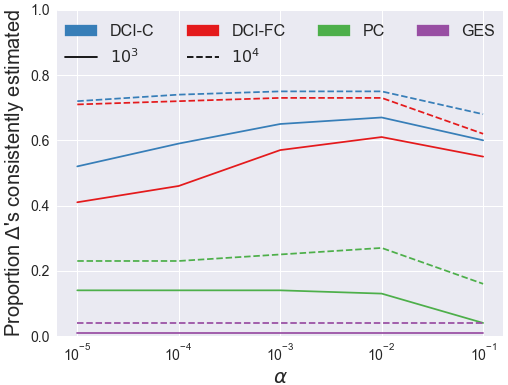}}
	\subfigure[changed variances]{\includegraphics[width=.4\textwidth]{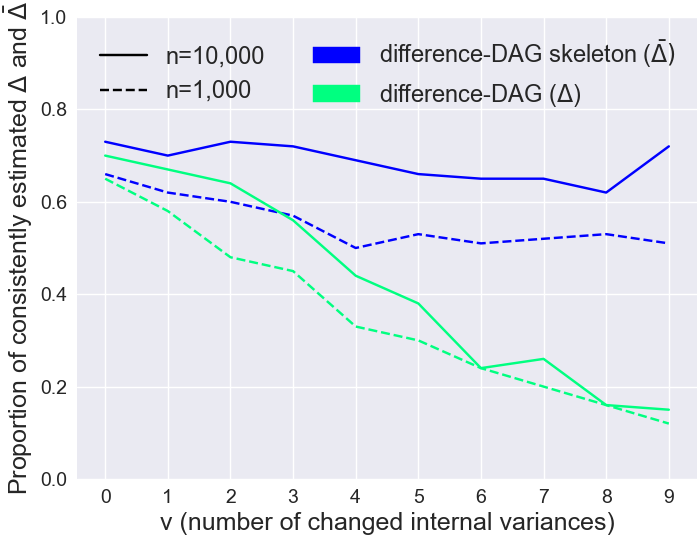}}
	\caption{Proportion of consistently estimated D-DAGs for 100 realizations per setting with $p = 10$ nodes and sample size $n$. Figures (a) and (b) show the proportion of consistently estimated D-DAGs when considering just the skeleton ($\bar{\Delta}$) and both skeleton and edge orientations ($\Delta$), respectively;  $\alpha$ is the significance level used for the hypothesis tests in the algorithms. Figure $(c)$ shows the proportion of consistent estimates with respect to the number of changes in internal node variances $v$.}
	\label{fig:low-dim}
	\vspace{-0.3cm}
\end{figure}

In Figure~\ref{fig:low-dim} we analyzed how the performance of the DCI algorithm changes over different choices of significance levels $\alpha$. 
The simulations were performed on graphs with $p=10$ nodes, neighborhood size of $s=3$ and sample size $n\in\{10^3, 10^4\}$.
For Figure~\ref{fig:low-dim} (a) and (b) we set $\epsilon^{(1)}, \epsilon^{(2)} \sim \mathcal{N}(0, {\bf 1}_p)$, which by Corollary~\ref{cor:dir} ensures that the D-DAG $\Delta$ is fully identifiable. We compared the performance of DCI to the naive approach, where we separately estimated the two DAGs $\G^{(1)}$ and $\G^{(2)}$ and then took their difference. For separate estimation we used the prominent  PC and GES algorithms tailored to the Gaussian setting. Since KLIEP requires an additional tuning parameter, to understand how $\alpha$ influences the performance of the DCI algorithm, we here only analyzed initializations in the fully connected graph (DCI-FC) and using the constraint-based method described in the Supplementary Material (DCI-C). Both initializations provide a provably consistent algorithm. 
Figure~\ref{fig:low-dim} (a) and (b) show the proportion of consistently estimated D-DAGs by just considering the skeleton ($\bar{\Delta}$) and both skeleton and orientations ($\Delta$), respectively. For PC and GES, we considered the set of edges that appeared in one estimated skeleton but disappeared in the other as the estimated skeleton of the D-DAG $\bar{\Delta}$. In determining orientations, we considered the arrows that were directed in one estimated CP-DAG but disappeared in the other as the estimated set of directed arrows. Since the main purpose of this low-dimensional simulation study is to validate our theoretical findings, we used the exact recovery rate as evaluation criterion. In line with our theoretical findings, both variants of the DCI algorithm outperformed taking differences after separate estimation. 
Figure~\ref{fig:low-dim} (a) and (b)  also show that the PC algorithm outperformed GES, which is unexpected given previous results showing that GES usually has a higher exact recovery rate than the PC algorithm for estimating a single DAG. This is due to the fact that while the PC algorithm usually estimates less DAGs correctly, the incorrectly estimated DAGs tend to look more similar to the true model than the incorrect estimates of GES (as also reported in~\cite{SWM17}) and can still lead to a correct estimate~of~the~D-DAG.

In Figure~\ref{fig:low-dim} (c) we analyzed the effect of changes in the noise variances on estimation performance. We set $\epsilon^{(1)} \sim \mathcal{N}(0, {\bf 1}_p)$, while for $\epsilon^{(2)}$ we randomly picked $v$ nodes and uniformly sampled their variances from $[1.25, 2]$. We used $\alpha=.05$ as significance level based on the evaluation from Figure~\ref{fig:low-dim}. 
In line with Theorem~\ref{thm:dir}, as we increase the number of nodes $i$ such that $\epsilon_i^{(1)} \neq \epsilon_i^{(2)}$,  the number of edges whose orientations can be determined decreases. This is because Algorithm~\ref{alg:dir} can only determine an edge's orientation when the variance of at least one of its nodes is invariant. Moreover, Figure~\ref{fig:low-dim}~(c) shows that the accuracy of Algorithm~\ref{alg:skel} is not impacted by changes in the noise variances.

Finally, Figure~\ref{fig:high-dim} (a) - (b) show the performance (using ROC curves) of the DCI algorithm in the high-dimensional setting when initiated using KLIEP (DCI-K) and DCI-C. The simulations were performed on graphs with $p=100$ nodes, expected neighborhood size of $s=10$, sample size $n=300$, and  $\epsilon^{(1)}, \epsilon^{(2)} \sim \mathcal{N}(0, {\bf 1}_p)$.
$\B2$ was derived from $\B1$ so that the total number of changes was $5\%$ of the total number of edges in $\B1$, with an equal amount of insertions and deletions. 
Figure~\ref{fig:high-dim} (a) - (b) show that both DCI-C and DCI-K 
perform similarly well and outperform separate estimation using GES and the PC algorithm. 
The respective plots for 10\% change between $\B1$ and $\B2$ are given in the Supplementary Material.

\begin{figure}[t]
	\centering
	\subfigure[D-DAG skeleton $\bar{\Delta}$]
	{\includegraphics[width=.25\textwidth]{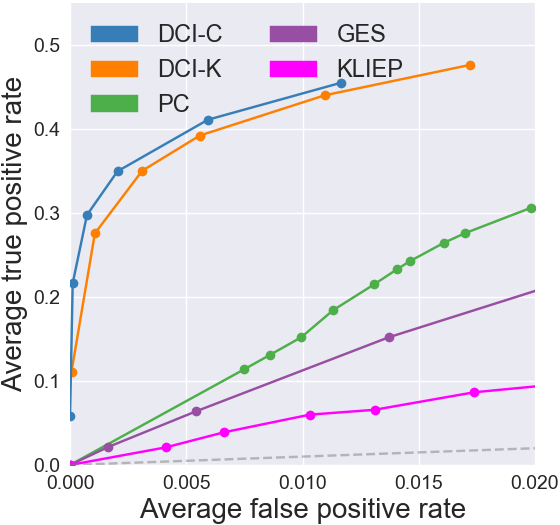}}
	\subfigure[D-DAG $\Delta$]
	{\includegraphics[width=.25\textwidth]{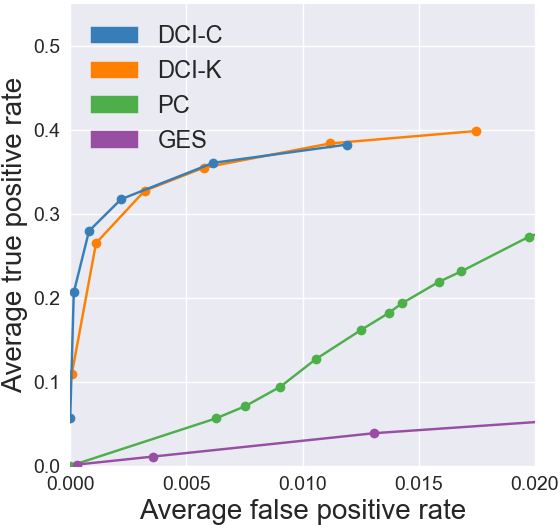}}
	\subfigure[T-cell activation]
	{\includegraphics[width=.47\textwidth]{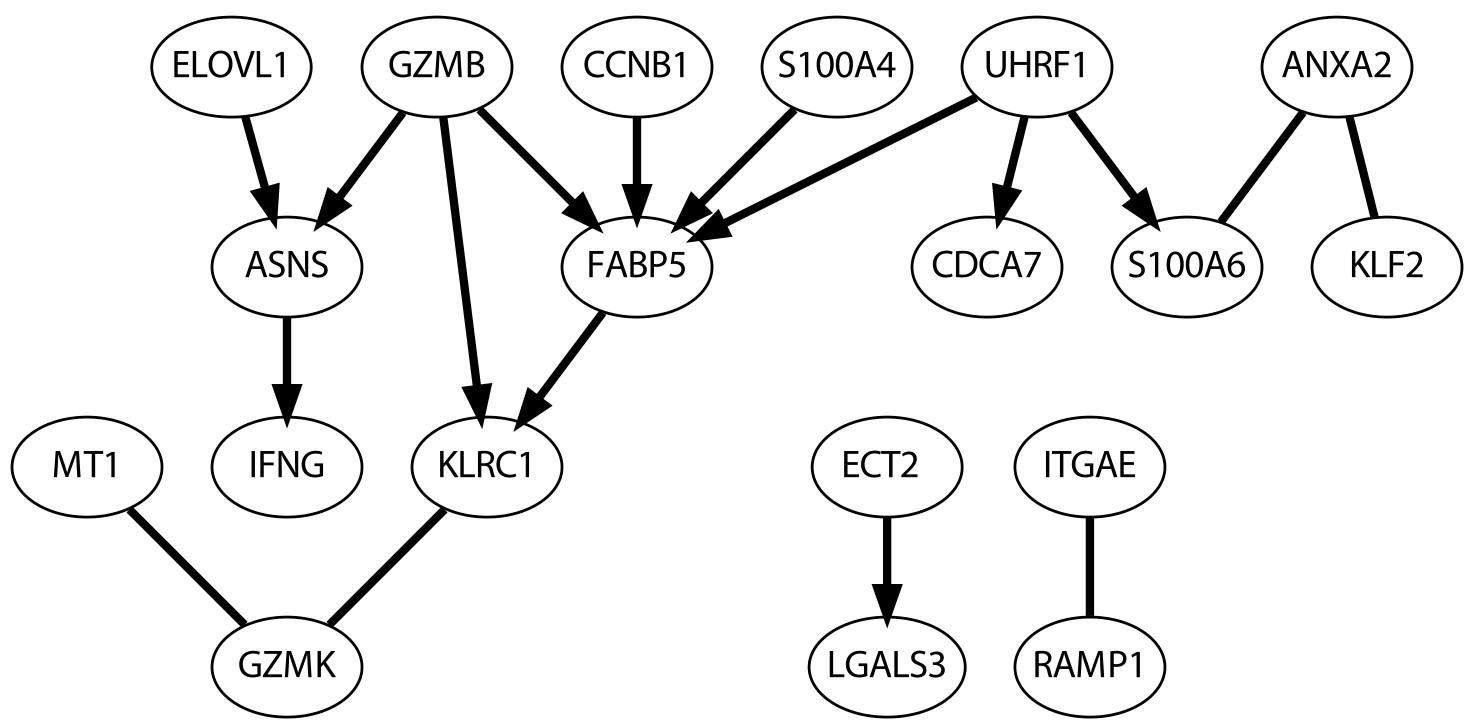}}
	\caption{High-dimensional evaluation of the DCI algorithm in both simulation and real data; $(a) - (b)$ are the ROC curves for estimating the D-DAG $\Delta$ and its skeleton $\bar{\Delta}$ with $p = 100$ nodes, expected neighbourhood size $s = 10$, $n = 300$ samples, and 5\% change~between~DAGs; $(c)$ shows the estimated D-DAG between gene expression data from naive and activated T cells.}
	\label{fig:high-dim}
	\vspace{-0.2cm}
\end{figure}


%

\subsection{Real data analysis}

\vspace{-0.2cm}


{\bf Ovarian cancer.}
We tested our method on an ovarian cancer data set~\cite{TTG08} that contains two groups of patients with different survival rates and was previously analyzed using the DPM algorithm in the undirected setting~\cite{ZCL14}. We followed the analysis of~\cite{ZCL14} and applied the DCI algorithm to gene expression data from the apoptosis and TGF-$\beta$ pathways. In the apoptosis pathway we identified two hub nodes: 
BIRC3, also discovered by DPM, is an inhibitor of apoptosis~\cite{Johnstone2008} and one of the main disregulated genes in ovarian cancer~\cite{Jonsson2014}; PRKAR2B, not identified by DPM, has been shown to be important in disease progression in ovarian cancer cells~\cite{CNW08} and an important regulatory unit for cancer cell growth~\cite{CBG08}. In addition, the RII-$\beta$ protein encoded by PRKAR2B has been considered as a therapeutic target for cancer therapy~\cite{CCY99,MGM06}, thereby confirming the relevance of our findings. With respect to the TGF-$\beta$ pathway, the DCI method identified THBS2 and COMP as hub nodes. Both of these genes have been implicated in resistance to chemotherapy in epithelial ovarian cancer~\cite{MFC13} and were also recovered by DPM. Overall, the D-UG discovered by DPM is comparable to the D-DAG found by our method. More details on this analysis are given in the Supplementary Material.

{\bf T cell activation.} 
To demonstrate the relevance of our method for current genomics applications, we applied DCI to single-cell gene expression data of naive and activated T cells in order to study the pathways involved 
during the immune response to a pathogen. We analyzed data from 377 activated and 298 naive T cells obtained by~\cite{Singer2016} using the recent drop-seq technology. From the previously identified differentially expressed genes between naive and activated T cells \cite{Sarkar2008}, we selected all genes that had a fold expression change above 10, resulting in 60 genes for further analysis.

We initiated DCI using KLIEP, thresholding the edge weights at 0.005, 
and ran DCI for different tuning parameters and with cross-validation to obtain the final DCI output shown in Figure~\ref{fig:high-dim} (c) using stability selection as described in~\cite{Meinshausen2010}. The genes with highest out-degree, and hence of interest for future interventional experiments, are GZMB and UHRF1. 
Interestingly, GZMB 
is known to induce cytotoxicity, important  for attacking and killing the invading pathogens. Furthermore, this gene has 
been reported as the most differentially expressed gene during T cell activation~\cite{Hatton2013,Peixoto2007}.
UHRF1 has been shown to be critical for T cell maturation and proliferation through knockout experiments~\cite{Cui2016,Obata2014}. Interestingly, the UHRF1 protein is a transcription factor, i.e.~it binds to DNA sequences and regulates the expression of other genes, thereby confirming its role as an important causal regulator. 
Learning a D-DAG as opposed to a D-UG is crucial for prioritizing interventional experiments. In addition, the difference UG for this application would not only have been more dense, but it would also have resulted in additional hub nodes such as FABP5, KLRC1, and ASNS, which based on the current biological literature seem secondary to T cell activation (FABP5 is involved in lipid binding, KLRC1 has a role in natural killer cells but not in T cells, and ASNS is an asparagine synthetase gene). 
The difference DAGs learned by separately applying the GES and PC algorithms on naive and activated T cell data sets as well as on the ovarian cancer data sets are included in the Supplementary Material for comparison.

\section{Discussion}\label{sec:discuss}

We presented an algorithm for directly estimating the difference between two causal DAG models given i.i.d.~samples from each model. To our knowledge this is the first such algorithm and is of particular interest for learning differences between related networks, where each network might be large and complex, while the difference is sparse. We provided consistency guarantees for our algorithm and showed on synthetic and real data that it outperforms the naive approach of separately estimating two DAG models and taking their difference. While our proofs were for the setting with no latent variables, they extend to the setting where the edge weights and noise terms of all latent variables remain invariant across the two DAGs. We applied our algorithm to gene expression data in bulk and from single cells, showing that DCI is able to identify biologically relevant genes for ovarian cancer and T-cell activation. This purports DCI as a promising method for identifying intervention targets that are causal for a particular phenotype for subsequent experimental validation. A more careful analysis with respect to the D-DAGs discovered by our DCI algorithm is needed to reveal its impact for scientific discovery.

In order to make DCI scale to networks with thousands of nodes, an important challenge is to reduce the number of hypothesis tests. As mentioned in Remark~\ref{rk:complexity}, currently the time complexity (given by the number of hypothesis tests) of DCI scales exponentially with respect to the size of $S_\Theta$. The PC algorithm overcomes this problem by dynamically updating the list of CI tests given the current estimate of the graph. It is an open problem whether one can similarly reduce the number of hypothesis tests for DCI. Another challenge is to relax Assumptions~4.1 and~4.2. Furthermore, in many applications (e.g., when comparing normal to disease states), there is an imbalance of data/prior knowledge for the two models and it is of interest to develop methods that can make use of this for learning the differences between the two models.


Finally, as described in Section~\ref{sec:related}, DCI is preferable to separate estimation methods like PC and GES since it can infer not only edges that appear or disappear, but also edges with changed edge weights. However, unlike separate estimation methods, DCI relies on the assumption that the two DAGs share a topological order. 
Developing methods to directly estimate the difference of two DAGs that do not share a topological order is of great interest for future work.

\section*{Acknowledgements}
Yuhao Wang was supported by ONR (N00014-17-1-2147), NSF (DMS-1651995) and the MIT-IBM Watson AI Lab. Anastasiya Belyaeva was supported by an NSF Graduate Research Fellowship (1122374) and the Abdul Latif Jameel World Water and Food Security Lab (J-WAFS) at MIT. Caroline Uhler was partially supported by ONR (N00014-17-1-2147), NSF (DMS-1651995), and a Sloan Fellowship.

\bibliography{nips_2018}

\begin{thebibliography}{10}

\bibitem{AMP97}
S.~A. Andersson, D.~Madigan, and M.~D. Perlman.
\newblock A characterization of markov equivalence classes for acyclic
  digraphs.
\newblock {\em The Annals of Statistics}, 25(2):505--541, 1997.

\bibitem{BGL11}
A.-L. Barab\'asi, N.~Gulbahce, and J.~Loscalzo.
\newblock Network medicine: a network-based approach to human disease.
\newblock {\em Nature Reviews Genetics}, 12(1):56--68, 2011.

\bibitem{BZ04}
A.-L. Barab\'asi and Z.~N. Oltvai.
\newblock Network biology: understanding the cell's functional organization.
\newblock {\em Nature Reviews Genetics}, 5(2):101--113, 2004.

\bibitem{CNW08}
C.~Cheadle, M.~Nesterova, T.~Watkins, K.~C. Barnes, J.~C. Hall, A.~Rosen, K.~G.
  Becker, and Y.~S. Cho-Chung.
\newblock Regulatory subunits of {PKA} define an axis of cellular
  proliferation/differentiation in ovarian cancer cells.
\newblock {\em BMC Medical Genomics}, 1(1):43, 2008.

\bibitem{CBG08}
F.~Chiaradonna, C.~Balestrieri, D.~Gaglio, and M.~Vanoni.
\newblock {RAS} and {PKA} pathways in cancer: new insight from transcriptional
  analysis.
\newblock {\em Frontiers in Bioscience}, 13:5257--5278, 2008.

\bibitem{CCY99}
Y.~S. Cho-Chung.
\newblock Antisense oligonucleotide inhibition of serine/threonine kinases: an
  innovative approach to cancer treatment.
\newblock {\em Pharmacology \& Therapeutics}, 82(2):437--449, 1999.

\bibitem{Cui2016}
Y.~Cui, X.~Chen, J.~Zhang, X.~Sun, H.~Liu, L.~Bai, C.~Xu, and X.~Liu.
\newblock {Uhrf1 Controls iNKT Cell Survival and Differentiation through the
  Akt-mTOR Axis}.
\newblock {\em Cell Reports}, 15(2):256--263, 2016.

\bibitem{DST13}
J.~Draisma, S.~Sullivant, and K.~Talaska.
\newblock Positivity for gaussian graphical models.
\newblock {\em Advances in Applied Mathematics}, 50(5):661--674, 2013.

\bibitem{DP07}
M.~Drton and M.~D. Perlman.
\newblock Multiple testing and error control in gaussian graphical model
  selection.
\newblock {\em Statistical Science}, pages 430--449, 2007.

\bibitem{FLN00}
N.~Friedman, M.~Linial, I.~Nachman, and D.~Pe'er.
\newblock Using bayesian networks to analyze expression data.
\newblock {\em Journal of Computational Biology}, 7(3-4):601--620, 2000.

\bibitem{GSKZ17}
A.~Ghassami, S.~Salehkaleybar, N.~Kiyavash, and K.~Zhang.
\newblock Learning causal structures using regression invariance.
\newblock In {\em Advances in Neural Information Processing Systems}, pages
  3015--3025, 2017.

\bibitem{Hatton2013}
L.~A. Hatton.
\newblock {\em Molecular Mechanisms Regulating {CD8+ T} Cell Granzyme and
  Perforin Gene Expression}.
\newblock PhD thesis, University of Melbourne, 2013.

\bibitem{HRD09}
N.~J. Hudson, A.~Reverter, and B.~P. Dalrymple.
\newblock A differential wiring analysis of expression data correctly
  identifies the gene containing the causal mutation.
\newblock {\em PLoS Computational Biology}, 5(5):e1000382, 2009.

\bibitem{Johnstone2008}
R.W. Johnstone, A.J. Frew, and M.J. Smyth.
\newblock {The TRAIL apoptotic pathway in cancer onset, progression and
  therapy}.
\newblock {\em Nature Reviews Cancer}, 8(10):782--798, 2008.

\bibitem{Jonsson2014}
J.~J{\"{o}}nsson, K.~Bartuma, M.~Dominguez-Valentin, K.~Harbst, Z.~Ketabi,
  S.~Malander, M.~J{\"{o}}nsson, A.~Carneiro, A.~M{\aa}sb{\"{a}}ck,
  G.~J{\"{o}}nsson, and M.~Nilbert.
\newblock {Distinct gene expression profiles in ovarian cancer linked to Lynch
  syndrome.}
\newblock {\em Familial Cancer}, 13:537--545, 2014.

\bibitem{KB07}
M.~Kalisch and P.~B{\"u}hlmann.
\newblock Estimating high-dimensional directed acyclic graphs with the
  pc-algorithm.
\newblock {\em Journal of Machine Learning Research}, 8(Mar):613--636, 2007.

\bibitem{KGS11}
M.~Kanehisa, S.~Goto, Y.~Sato, M.~Furumichi, and T.~Mao.
\newblock {KEGG} for integration and interpretation of large-scale molecular
  data sets.
\newblock {\em Nucleic Acids Research}, 40(D1):D109--D114, 2011.

\bibitem{LAU96}
S.~L Lauritzen.
\newblock {\em Graphical Models}, volume~17.
\newblock Clarendon Press, 1996.

\bibitem{LUS14}
S.~Lin, C.~Uhler, B.~Sturmfels, and P.~B{\"u}hlmann.
\newblock Hypersurfaces and their singularities in partial correlation testing.
\newblock {\em Foundations of Computational Mathematics}, 14(5):1079--1116,
  2014.

\bibitem{LFS17}
S.~Liu, K.~Fukumizu, and T.~Suzuki.
\newblock Learning sparse structural changes in high-dimensional markov
  networks.
\newblock {\em Behaviormetrika}, 44(1):265--286, 2017.

\bibitem{LQG14}
S.~Liu, J.~A. Quinn, M.~U. Gutmann, T.~Suzuki, and M.~Sugiyama.
\newblock Direct learning of sparse changes in markov networks by density ratio
  estimation.
\newblock {\em Neural Computation}, 26(6):1169--1197, 2014.

\bibitem{LUT05}
Helmut L{\"u}tkepohl.
\newblock {\em New introduction to multiple time series analysis}.
\newblock Springer Science \& Business Media, 2005.

\bibitem{MFC13}
S.~Marchini, R.~Fruscio, L.~Clivio, L.~Beltrame, L.~Porcu, I.~F. Nerini,
  D.~Cavalieri, G.~Chiorino, G.~Cattoretti, C.~Mangioni, R.~Milani, V.~Torri,
  C.~Romualdi, A.~Zambelli, M.~Romano, M.~Signorelli, S.~di~Giandomenico, and
  M~D'Incalci.
\newblock Resistance to platinum-based chemotherapy is associated with
  epithelial to mesenchymal transition in epithelial ovarian cancer.
\newblock {\em European Journal of Cancer}, 49(2):520--530, 2013.

\bibitem{Meek1997}
C.~Meek.
\newblock {\em Graphical Models: Selecting Causal and Statistical Models}.
\newblock PhD thesis, Carnegie Mellon University, 1997.

\bibitem{Meinshausen2010}
N.~Meinshausen and P.~B{\"{u}}hlmann.
\newblock {Stability selection}.
\newblock {\em Journal of the Royal Statistical Society. Series B: Statistical
  Methodology}, 72(4):417--473, 2010.

\bibitem{MGM06}
T.~Mikalsen, N.~Gerits, and U.~Moens.
\newblock Inhibitors of signal transduction protein kinases as targets for
  cancer therapy.
\newblock {\em Biotechnology Annual Review}, 12:153--223, 2006.

\bibitem{NHM15}
P.~Nandy, A.~Hauser, and M.~H. Maathuis.
\newblock High-dimensional consistency in score-based and hybrid structure
  learning, 2015.
\newblock To appear in \emph{Annals of Statistics}.

\bibitem{Obata2014}
Y.~Obata, Y.~Furusawa, T.A. Endo, J.~Sharif, D.~Takahashi, K.~Atarashi,
  M.~Nakayama, S.~Onawa, Y.~Fujimura, M.~Takahashi, T.~Ikawa, T.~Otsubo, Y.I.
  Kawamura, T.~Dohi, S.~Tajima, H.~Masumoto, O.~Ohara, K.~Honda, S.~Hori,
  H.~Ohno, H.~Koseki, and K.~Hase.
\newblock {The epigenetic regulator Uhrf1 facilitates the proliferation and
  maturation of colonic regulatory T cells}.
\newblock {\em Nature Immunology}, 15(6):571--579, 2014.

\bibitem{OGS99}
H.~Ogata, S.~Goto, K.~Sato, W.~Fujibuchi, H.~Bono, and M.~Kanehisa.
\newblock {KEGG}: Kyoto encyclopedia of genes and genomes.
\newblock {\em Nucleic Acids Research}, 27(1):29--34, 1999.

\bibitem{Pearl:00}
J.~Pearl.
\newblock {\em Causality: Models, Reasoning, and Inference}.
\newblock Cambridge University Press, 2000.

\bibitem{Peixoto2007}
A.~Peixoto, C.~Evaristo, I.~Munitic, M.~Monteiro, A.~Charbit, B.~Rocha, and
  H.~Veiga-Fernandes.
\newblock {CD8 single-cell gene coexpression reveals three different effector
  types present at distinct phases of the immune response}.
\newblock {\em The Journal of Experimental Medicine}, 204(5):1193--1205, 2007.

\bibitem{PBM16}
J.~Peters, P.~B{\"u}hlmann, and N.~Meinshausen.
\newblock Causal inference by using invariant prediction: identification and
  confidence intervals.
\newblock {\em Journal of the Royal Statistical Society: Series B (Statistical
  Methodology)}, 78(5):947--1012, 2016.

\bibitem{POK07}
J.~E. Pimanda, K.~Ottersbach, K.~Knezevic, S.~Kinston, W.~Y. Chan, N.~K.
  Wilson, J.~Landry, A.~D Wood, A.~Kolb-Kokocinski, A.~R. Green, D.~Tannahill,
  G.~Lacaud, V.~Kouskoff, and B.~G\"{o}ttgens.
\newblock Gata2, {Fli1}, and {Scl} form a recursively wired gene-regulatory
  circuit during early hematopoietic development.
\newblock {\em Proceedings of the National Academy of Sciences},
  104(45):17692--17697, 2007.

\bibitem{POU11}
M.~Pourahmadi.
\newblock Covariance estimation: The glm and regularization perspectives.
\newblock {\em Statistical Science}, pages 369--387, 2011.

\bibitem{RSZ06}
J.~Ramsey, P.~Spirtes, and J.~Zhang.
\newblock Adjacency-faithfulness and conservative causal inference.
\newblock In {\em Proceedings of the Twenty-Second Conference on Uncertainty in
  Artificial Intelligence}, pages 401--408. AUAI Press, 2006.

\bibitem{RHB00}
J.~M. Robins and B.~Hernan, Miguel.~A.and~Brumback.
\newblock Marginal structural models and causal inference in epidemiology,
  2000.

\bibitem{SC13}
S.~Sanei and J.~A. Chambers.
\newblock {\em EEG Signal Processing}.
\newblock John Wiley \& Sons, 2013.

\bibitem{Sarkar2008}
S.~Sarkar, V.~Kalia, W.N. Haining, B.T. Konieczny, S.~Subramaniam, and
  R.~Ahmed.
\newblock {Functional and genomic profiling of effector CD8 T cell subsets with
  distinct memory fates}.
\newblock {\em The Journal of Experimental Medicine}, 205(3):625--640, 2008.

\bibitem{SHHK06}
S.~Shimizu, P.~O. Hoyer, A.~Hyv{\"a}rinen, and A.~Kerminen.
\newblock A linear non-gaussian acyclic model for causal discovery.
\newblock {\em Journal of Machine Learning Research}, 7(Oct):2003--2030, 2006.

\bibitem{Singer2016}
M.~Singer, C.~Wang, L.~Cong, N.D. Marjanovic, M.S. Kowalczyk, H.~Zhang,
  J.~Nyman, K.~Sakuishi, S.~Kurtulus, D.~Gennert, J.~Xia, J.Y.H. Kwon,
  J.~Nevin, R.H. Herbst, I.~Yanai, O.~Rozenblatt-Rosen, V.K. Kuchroo, A.~Regev,
  and A.C. Anderson.
\newblock {A Distinct Gene Module for Dysfunction Uncoupled from Activation in
  Tumor-Infiltrating T Cells}.
\newblock {\em Cell}, 166(6):1500--1511, 2016.

\bibitem{SWM17}
L.~Solus, Y.~Wang, L.~Matejovicova, and C.~Uhler.
\newblock Consistency guarantees for permutation-based causal inference
  algorithms, 2017.

\bibitem{SGS00}
P.~Spirtes, C.~N. Glymour, and R.~Scheines.
\newblock {\em Causation, Prediction, and Search}.
\newblock MIT press, 2000.

\bibitem{TTG08}
R.~W. Tothill, A.~V. Tinker, J.~George, R.~Brown, S.~B. Fox, S.~Lade, D.~S.
  Johnson, M.~K. Trivett, D.~Etemadmoghadam, B.~Locandro, N.~Traficante,
  S.~Fereday, J.~A. Hung, Y.~Chiew, I.~Haviv, Australian Ovarian Cancer~Study
  Group, D.~Gertig, A.~deFazio, and D.~D.L. Bowtell.
\newblock Novel molecular subtypes of serous and endometrioid ovarian cancer
  linked to clinical outcome.
\newblock {\em Clinical Cancer Research}, 14(16):5198--5208, 2008.

\bibitem{TBA06}
I.~Tsamardinos, L.~E. Brown, and C.~F. Aliferis.
\newblock The max-min hill-climbing bayesian network structure learning
  algorithm.
\newblock {\em Machine Learning}, 65(1):31--78, 2006.

\bibitem{URB13}
C.~Uhler, G.~Raskutti, P.~B{\"u}hlmann, and B.~Yu.
\newblock Geometry of the faithfulness assumption in causal inference.
\newblock {\em The Annals of Statistics}, pages 436--463, 2013.

\bibitem{GB13}
S.~Van~de Geer and P.~B{\"u}hlmann.
\newblock $\ell_0$-penalized maximum likelihood for sparse directed acyclic
  graphs.
\newblock {\em The Annals of Statistics}, 41(2):536--567, 2013.

\bibitem{VP90}
T.~Verma and J.~Pearl.
\newblock Equivalence and synthesis of causal models.
\newblock In {\em Proceedings of the Sixth Annual Conference on Uncertainty in
  Artificial Intelligence}, pages 255--270. Elsevier Science Inc., 1990.

\bibitem{ZHZ17}
K.~Zhang, B.~Huang, J.~Zhang, C.~Glymour, and B.~Sch{\"o}lkopf.
\newblock Causal discovery from nonstationary/heterogeneous data: Skeleton
  estimation and orientation determination.
\newblock In {\em IJCAI: Proceedings of the Conference}, volume 2017, page
  1347. NIH Public Access, 2017.

\bibitem{ZCL14}
S.~D. Zhao, T.~T. Cai, and H.~Li.
\newblock Direct estimation of differential networks.
\newblock {\em Biometrika}, 101(2):253--268, 2014.

\end{thebibliography}
\bibliographystyle{plain}

\newpage

\appendix

\newcommand{\snum}{S}
\renewcommand{\theequation}{\snum.\arabic{equation}}
\counterwithin{algorithm}{section}
\counterwithin{figure}{section}

\section{Hypothesis testing framework}
 In this section, we provide the details regarding the hypothesis tests that we used for testing the following two null hypotheses:
\begin{align*}
H_{0}^{i,j\mid S}: \beta_{i,j \mid S}^{(1)} = \beta_{i,j \mid S}^{(2)} \qquad \textrm{and} \qquad 
H_{0}^{j\mid S}: \sigma_{j \mid S}^{(1)} = \sigma_{j \mid S}^{(2)}.
\end{align*}
As in~\cite{GSKZ17}, we used hypothesis tests based on the F-test for testing for invariance between regression coefficients and residual variances. For testing $H_{0}^{i,j\mid S}: \beta_{i,j \mid S}^{(1)} = \beta_{i,j \mid S}^{(2)}$ we used the test statistic
\begin{align*}
\hat{T} \!\!:= (\hat{\beta}_{i,j \mid S}^{(1)} - \hat{\beta}_{i,j \mid S}^{(2)})^2 \cdot \quad \big[\big((\hat{\sigma}_{j \mid M}^{(1)}\!)^2(n_1\!\hat{\Sigma}_{M,M}^{(1)}\!)^{-1} \!\!\!+\! (\hat{\sigma}_{j \mid M}^{(2)}\!)^2(n_2\hat{\Sigma}_{M,M}^{(2)}\!)^{-1}\big)^{-1}\big]_{i_M i_M}
\end{align*}
where $\hat{\beta}_{i,j \mid S}^{(k)}$ is the empirical estimate of $\beta_{i,j \mid S}^{(k)}$ obtained by ordinary least squares, $(\hat{\sigma}_{j \mid M}^{(k)})^2$ is an unbiased estimator of the regression residual variance $(\sigma_{j \mid M}^{(k)})^2$,  $\hat{\Sigma}_{M,M}^{(k)}$ is the sample covariance matrix of the random vector $X_M^{(k)}$ with $M = \{i\} \cup S$, and $i_M$ denotes the index in $M$ corresponding to the element $i$. In~\cite{LUT05}[Section~3.6], the author shows that under the null hypothesis the asymptotic distribution of $\hat{T}$ can be approximated by the F-distribution $F(1, n_1 + n_2 - 2 | S | - 2)$. The basic explanation is that, for $M :=S \cup \{i\}$, let $\beta_{M}^{(k)}$ be the best linear predictor when regressing $\Xk_j$ onto $\Xk_{M}$, i.e., our estimator is $\Xk_j = (\beta_{M}^{(k)})^T \Xk_{M} + \tilde{\epsilon}_j^{(k)}$. Let $\beta$ be the vector 
$$\beta := \left[ \begin{matrix} \beta_{M}^{(1)} \\ \beta_{M}^{(2)} \end{matrix}\right],$$
and let $C \in \R^{2|M|}$ have $C_{i_{M}} = 1$ and $C_{|M| + i_{M}} = -1$ and all other entries as zero. Then the null hypothesis $H_{0}^{i,j|S}$ can be written as: $C^T \beta = 0$. It follows from Proposition~3.5 of~\cite{LUT05}, on the asymptotic distribution of the Wald statistic, that $\hat{T}$ converges in distribution to $\chi^2(1)$, i.e., a $\chi^2$-distribution with $1$ degree of freedom. 

However, the F-distribution $F(1, n_1 + n_2 - 2 | S | - 2)$ is a better approximation for the distribution of $\hat{T}$, as outlined in Section~3.6 of~\cite{LUT05}. A brief justification is in order. First, we know that the convergence is the same: for an F-distribution $F(1, d)$, as $d \to \infty$, we have $F(1, d) \overset{d}{\to} \chi^2(1)$. Additionally, $F(1, d)$ and $\hat{T}$ both have a fatter tail than $\chi^2(1)$. Together, these facts suggest the choice of a F-distribution $F(1, d)$ with $d \to \infty$ as $n_1 , n_2 \to \infty$. For the second parameter $d$, we used $d = n_1 + n_2 - 2 | S | - 2$, the total degrees of freedom of the unbiased estimators of the two regression residual variances, i.e., $(\hat{\sigma}_{j \mid M}^{(1)})^2$ and $(\hat{\sigma}_{j \mid M}^{(2)})^2$.

Similarly, for testing $H_{0}^{j\mid S}$, we used the test statistic $$\hat{F} := (\hat{\sigma}_{j \mid S}^{(1)})^2 / (\hat{\sigma}_{j \mid S}^{(2)})^2.$$ Under the null hypothesis, $\hat{F}$ is a ratio of two $\chi^2$-distributed random variables and hence $\hat{F}$ follows an F-distribution, namely $F(n_1 - | S | - 1, n_2 - | S | - 1)$. 

\section{Comparison to related work on invariant causal structure learning}

The complimentary problem to learning the difference of two DAG models is the problem of inferring the causal structure that is \emph{invariant} across different environments. Algorithms for this problem have been developed in recent literature~\cite{GSKZ17,PBM16,ZHZ17}. Since the hypothesis testing framework in~\cite{GSKZ17} is similar to our approach, we here provide an example to explain the differences between the two approaches and in particular to show that a new approach is needed in order to obtain a consistent method for learning the difference DAG.  


Recall that when we have access to data from a pair of DAGs, the algorithms in~\cite{GSKZ17} make use of the following two sets that are estimated from the data. The first is the {\it regression invariance set}: 
\begin{align*}
R := \Big\{(j, S) : \beta_S^{(1)}(j) = \beta_S^{(2)}(j)\Big\},
\end{align*}
where $\beta_S^{(k)}(j)$ corresponds to the best linear predictor when regressing $X_j^{(k)}$ onto $X_S^{(k)}$. The second is $I$, the set of variables whose internal noise variances have been changed across the two DAGs:
\begin{align*}
I := \Big\{j : \forall S \subseteq [p] \setminus \{j\}, \mathbb{E}(X_j^{(1)} - (\beta_S^{(1)}(j))^T X_S^{(1)})^2 \neq \mathbb{E}(X_j^{(2)} - (\beta_S^{(2)}(j))^T X_S^{(2)})^2\Big\}.
\end{align*}
The output of the algorithms in~\cite{GSKZ17} is fully determined by the invariant elements given in $R$ and $I$. In particular, Algorithm~1 in~\cite{GSKZ17} estimates the invariant causal structure by considering all elements in $R$ and $I$, while Algorithm~2 in~\cite{GSKZ17} is a more efficient constraint-based algorithm that considers only a subset of the elements in $R$. 

\begin{figure}[!b]
	\centering
	\subfigure[]{\includegraphics[width=.25\textwidth]{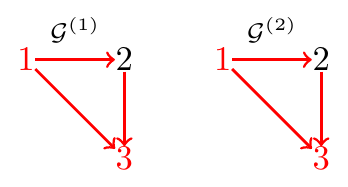}} \hspace{1.5cm}
	\subfigure[]{\includegraphics[width=.25\textwidth]{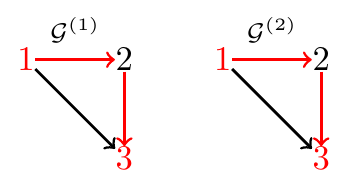}} \\
	\subfigure[]{\includegraphics[width=.35\textwidth]{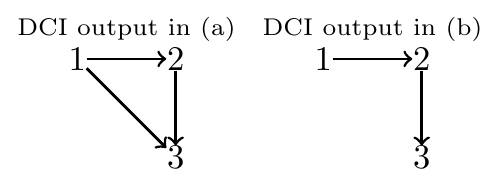}}
	\caption{(a) - (b): Example of two DAG pairs where the corresponding D-DAGs are different but the application of algorithms~1 and~2 from~\cite{GSKZ17} would result in the same sets $R$ and $I$. 
		The red edges correspond to the edges that have different edge weights across the two DAGs, the black edges correspond to the edges that have the same edge weights across the two DAGs. The red nodes correspond to the nodes that have different internal noise variances across the two DAGs and the black nodes correspond to the nodes that have unchanged internal noise variances. (c): D-DAGs output by the DCI algorithm when data is generated from (a) and (b), respectively.}
	\label{fig:ghassami}
	\vspace{-0.4cm}
\end{figure}

\begin{example}
	Figure~\ref{fig:ghassami} shows two cases where the underlying D-DAGs are different but in both cases~\cite{GSKZ17} would produce the same sets $R$ and $I$ that are used to assign edge orientations. In (a) we consider two fully connected linear SEMs $(\B1, \epsilon^{(1)})$ and $(\B2, \epsilon^{(2)})$ where the edge weights of all edges change across the two DAGs. The variances of the internal noise terms for nodes $1$ and $3$ change while the variance of the internal noise term of node $2$ stays the same. In (b) we instead consider  two fully connected linear SEMs $(\B1, \epsilon^{(1)})$ and $(\B2, \epsilon^{(2)})$ where $\B1_{12} \neq \B2_{12}$ and $\B1_{23} \neq \B2_{23}$. Moreover, the variances of the internal noise terms of nodes $1$ and $3$ change across $k = \{1,2\}$ while the variance of node $2$ stays the same. It can be easily shown that in both cases $R = \emptyset$ and $I = \{1,3\}$. Since both (a) and (b) correspond to exactly the same $R$ and $I$, by simply using the output from~\cite{GSKZ17}, we cannot distinguish whether the data is generated from the pair of DAGs given in (a) or the pair of DAGs given in (b). In fact, since for these examples $R$ is empty,~\cite{GSKZ17} will not uncover any edge orientations consistent with the underlying DAGs $\G^{(1)}$ or $\G^{(2)}$.
	On the other hand, our algorithm is able to distinguish these two cases as well as discover the edge orientations of the underlying D-DAGs, as shown in Figure~\ref{fig:ghassami} (c).
\end{example}

\section{Theoretical analysis}


\subsection{Preliminaries: Schur complement} \label{sec:schur}

In this section, we 
describe how to use Schur complements to express $\beta_{i,j \mid S}^{(k)}$ and $(\sigma_{j \mid S}^{(k)})^2$ as rational functions in the variables $(\Bk_{ij})_{(i,j) \in A^{(k)}}$ and $(\sigma_j^{(k)})_{j \in [p]}$. This will be used for the proofs of Theorems~\ref{thm:skel} and~\ref{thm:dir} in the following sections.

For a subset of nodes $M \subseteq [p]$, let $X_M$ denote the random vector spanned by the random variables $X_i$ for all $i \in M$. Let $\neg M$ denote the complement of $M$ with respect to the full set of nodes, i.e., $\neg M := [p] \setminus M$. The inverse covariance matrix of the random vector $X_M$, i.e., $(\Sigma_{M, M})^{-1}$, can be obtained from $\Theta$ by taking the Schur complement:
\begin{align*}
\begin{split}
& \Theta_M := (\Sigma_{M, M})^{-1} \\
&\quad = \Theta_{M,M} - \Theta_{M, \neg M} (\Theta_{\neg M, \neg M})^{-1} \Theta_{\neg M, M}.
\end{split}
\end{align*}
Note that here $\Theta_M$ does not represent the submatrix of $\Theta$ with set of row and column indices in $M$, i.e., $\Theta_{M,M}$, but rather the Schur complement. For any two indices $i, j \in M$, let $i_M,j_M \in \left[| M | \right]$ denote the row/column indices of matrix $\Theta_M$ associated to $i$ and $j$, then the $(i_M,j_M)$-th entry of matrix $\Theta_M$ can be written as:
\begin{align*}
(\Theta_M)_{i_M j_M} = \Theta_{ij} - \Theta_{i, \neg M} (\Theta_{\neg M, \neg M})^{-1} \Theta_{\neg M, j}.
\end{align*}
In \cite{DST13,LUS14,URB13} the authors also give a combinatorial characterization of the Schur complement. Following their characterization, the value of $(\Theta_M)_{i_M j_M}$ is determined by the parameters of the {\it d-connecting paths} from node $i$ to $j$ given $M \setminus \{i, j\}$. In this case, the entry $(\Theta_M^{(k)})_{i_M j_M}$ would be invariant for $k = \{1,2\}$ if the parameters along the d-connecting paths are all the same. Finally, by applying the result of~\cite{POU11}, $\beta_{i,j\mid S}^{(k)}$ and $(\sigma_{j \mid S}^{(k)})^2$ can be written as:
\begin{align}
\begin{split} \label{eq:inv}
\beta_{i,j\mid S}^{(k)} = - \frac{(\Theta_M^{(k)})_{i_M j_M}}{(\Theta_M^{(k)})_{j_M j_M}} \quad & \textrm{where} \quad M = S \union \{i, j\}, \\
(\sigma_{j \mid S}^{(k)})^2 = \left((\Theta_M^{(k)})_{j_M j_M} \right)^{-1} \quad & \textrm{where} \quad M = S \union \{j\}.
\end{split}
\end{align}
Combining Eq.~\eqref{eq:inv} with the formula for the Schur complement, one can easily see that $\beta_{i,j\mid S}^{(k)}$ and $(\sigma_{j \mid S}^{(k)})^2$ can be expressed as rational functions in the variables $(\Bk_{ij})_{(i,j) \in A^{(k)}}$ and $(\sigma_j^{(k)})_{j \in [p]}$.

\subsection{Proof of Theorem~\ref{thm:skel}} \label{sec:pfskel}

In this Section we provide the consistency proofs of Theorem~\ref{thm:skel} when Algorithm~\ref{alg:skel} is initialized in the difference-UG. The proof of Theorem~\ref{thm:skel} when Algorithm~\ref{alg:skel} is initialized in the complete graph follows easily from the proofs in this section. To complete the proof, one also needs the following assumption:

\begin{assumption}[Difference-precision-matrix-faithfulness assumption] \label{ass:prec}
For any choices of $i, j \in [p]$, it holds that
\begin{enumerate}
\item If $\B1_{ij} \neq \B2_{ij}$, then $\Theta_{ij}^{(1)} \neq \Theta_{ij}^{(2)}$, and for any $\ell$ with directed path $i \rightarrow j \leftarrow \ell$ in either $\G^{(1)}$ or $\G^{(2)}$, $\ell \in S_\Theta$.
\item If $\sigma_j^{(1)} \neq \sigma_j^{(2)}$, then $\Theta_{jj}^{(1)} \neq \Theta_{jj}^{(2)}$, and $~\forall~ i \in \pa^{(1)}(j) \cup \pa^{(2)}(j), i \in S_\Theta$.
\end{enumerate}
\end{assumption}

Note that Assumption~\ref{ass:prec} is not a necessary assumption for the consistency of Algorithm~\ref{alg:skel}, since one can simply take $\Delta_\Theta$ as the fully connected graph on $p$ nodes and $S_\Theta = [p]$ as input. The same holds for the proof of Theorem~\ref{thm:dir}. The strength of Assumption ~\ref{ass:prec} is further analyzed in Remark~\ref{rk:prec}.

To prove Theorem~\ref{thm:skel}, we need to make use of the following two lemmas:
\begin{lemma} \label{lem:skel1}
Given $\Theta^{(1)}$ and $\Theta^{(2)}$, if $\Theta^{(1)}_{ij} = \Theta^{(2)}_{ij}$, then $(\Theta^{(1)}_M)_{i_M j_M} = (\Theta^{(2)}_M)_{i_M j_M}$ for $M =S_\Theta \cup  \{ i, j \}$.
\end{lemma}
\begin{proof}
By Schur complement, we have that $(\Theta^{(k)}_M)_{i_M j_M} = \Theta_{ij}^{(k)} - \Theta_{i,\neg M}^{(k)} (\Theta_{\neg M, \neg M}^{(k)})^{-1} \Theta_{\neg M, j}^{(k)}$. By the definition of $S_\Theta$, $\Theta_{M, \neg M}^{(1)} = \Theta_{M, \neg M}^{(2)}$ and $\Theta_{\neg M, \neg M}^{(1)} = \Theta_{\neg M, \neg M}^{(2)}$.
\end{proof} 

\begin{lemma} \label{lem:skel3}
Given two linear SEMs $(\B1, \epsilon^{(1)})$ and $(\B2, \epsilon^{(2)})$ and denoting the precision precision matrix of the random vector $\Xk_{1:j}$ by $\Theta^{*(k)}$, then under Assumption~\ref{ass:prec} we have $S_{\Theta^*} \subseteq S_\Theta$.
\end{lemma}

\begin{proof}
Since $\Bk$ is strictly upper triangular, the marginal distribution of the random vector $\Xk_{1:j}$ follows a new SEM,
\begin{align*}
\Xk_{1:j} = (\Bk_{1:j,1:j})^T \Xk_{1:j} + \epsilon_{1:j}^{(k)},
\end{align*}
where $\Bk_{1:j, 1:j}$ is the submatrix of $\Bk$ with the first $j$ rows and $j$ columns, and $\epsilon_{1:j}^{(k)}$ is the random vector with the first $j$ random variables of $\epsilon^{(k)}$. It can then be shown that the $(i,\ell)$-th entry of the new precision matrix $\Theta^*$ is given by:
\begin{align*}
\Theta_{i \ell}^{*(k)} = - (\sigma_{\ell}^{(k)})^{-2} \Bk_{i\ell} + \sum\limits_{\ell < m \leq j} (\sigma_m^{(k)})^{-2} \Bk_{i m} \Bk_{\ell m}.
\end{align*}
It is then a short exercise to show that $\Theta_{i \ell}^{*(1)} \neq \Theta_{i \ell}^{*(2)}$ only if at least one of the following two statements hold:
\begin{enumerate}
\item $\B1_{i\ell} \neq \B2_{i\ell}$ or $\sigma_{\ell}^{(1)} \neq \sigma_{\ell}^{(2)}$;
\item There exists at least one of $\ell < m \leq j$ with $i \rightarrow m \leftarrow \ell$ in either $\G^{(1)}$ or $\G^{(2)}$ such that $\B1_{i m} \neq \B2_{i m}$ or $\B1_{\ell m} \neq \B2_{\ell m}$ or $\sigma_m^{(1)} \neq \sigma_m^{(2)}$.
\end{enumerate}
By applying Assumption~\ref{ass:prec}, we have that $\Theta_{i \ell}^{*(1)} \neq \Theta_{i \ell}^{*(2)} \Rightarrow i, \ell \in S_{\Theta}$.

The diagonal entries of the precision matrix are given by:
\begin{align*}
\Theta_{i i}^{*(k)} = (\sigma_i^{(k)})^{-2} + \sum\limits_{i < m \leq j} (\sigma_m^{(k)})^{-2} \Bk_{i m}.
\end{align*}
Clearly, $\Theta_{i i}^{*(1)} \neq \Theta_{i i}^{*(2)}$ only if at least one of the following statements hold:
\begin{enumerate}
\item $\sigma_i^{(1)} \neq \sigma_i^{(2)}$;
\item $\B1_{i m} \neq \B2_{i m}$ or $\sigma_m^{(1)} \neq \sigma_m^{(2)}$ for at least one of the descendents of $i$ in either $\G^{(1)}$ or $\G^{(2)}$ with $i < m \leq j$.
\end{enumerate}
By applying Assumption~\ref{ass:prec} we have that $\Theta_{i i}^{*(1)} \neq \Theta_{i i}^{*(2)} \Rightarrow i \in S_{\Theta}$.
\end{proof}

\begin{lemma} \label{lem:skel2}
Given two linear SEMs $(\B1, \epsilon^{(1)})$ and $(\B2, \epsilon^{(2)})$, then under Assumption~\ref{ass:adj_faith},  $\B1_{ij} = \B2_{ij}$ if and only if
\begin{align*}
\exists S \subseteq S_\Theta \setminus \{i,j\} \;\st\; \beta_{i,j\mid S}^{(1)} = \beta_{i,j\mid S}^{(2)} \;\textrm{or}\; \beta_{j,i\mid S}^{(1)} = \beta_{j,i\mid S}^{(2)}.
\end{align*}
\end{lemma}

\begin{proof}
We show the ``if'' direction by proving the contrapositive, i.e. if $\B1_{ij} \neq \B2_{ij}$, then 
\begin{align} \label{eq:thm1pf1}
\forall S \subseteq S_\Theta \setminus \{i,j\},\; \beta_{i,j\mid S}^{(1)} \neq \beta_{i,j\mid S}^{(2)} \;\textrm{and}\; \beta_{j,i\mid S}^{(1)} \neq \beta_{j,i\mid S}^{(2)}.
\end{align}
This follows directly from Assumption~\ref{ass:adj_faith}. 

Now, we prove the ``only if'' direction, i.e., if $\B1_{ij} = \B2_{ij}$, then 
\begin{align*}
\exists S \subseteq S_\Theta \setminus \{i,j\} \;\st\; \beta_{i,j\mid S}^{(1)} = \beta_{i,j\mid S}^{(2)} \;\textrm{or}\; \beta_{j,i\mid S}^{(1)} = \beta_{j,i\mid S}^{(2)}.
\end{align*}
We divide the proof into two cases: $\sigma_j^{(1)} = \sigma_j^{(2)}$, and $\sigma_j^{(1)} \neq \sigma_j^{(2)}$.\\
\textbf{Case 1} $\sigma_j^{(1)} = \sigma_j^{(2)}$

Consider the precision matrix $\Theta^{*(k)}$ of the random vector $X_{1:j}^{(k)}$. In this case, we prove that  choosing the conditioning set $S = S_{\Theta^*} \setminus \{i, j\}$ implies regression invariance. This is a valid choice for $S$, since it is a subset of $S_\Theta \setminus \{i, j\}$ by Lemma~\ref{lem:skel3}.

We will first show that $\Theta^{*(1)}_{ij} = \Theta^{*(2)}_{ij}$ and $\Theta^{*(1)}_{jj} = \Theta^{*(2)}_{jj}$. According to the new SEM of the marginal distribution of the random vector $\Xk_{1:j}$, i.e.,
\begin{align*}
\Xk_{1:j} = (\Bk_{1:j,1:j})^T \Xk_{1:j} + \epsilon_{1:j}^{(k)},
\end{align*}
it is easy to conclude that node $j$ no longer has any descendants in the marginal SEM. We therefore have that
\[
\Theta_{ij}^{*(k)} = -(\sigma_j^{(k)})^{-2} \Bk_{ij} \quad \textrm{and} \quad
\Theta_{jj}^{*(k)} = (\sigma_j^{(k)})^{-2}.
\]
Since $\B1_{ij} = \B2_{ij}$ and $\sigma_j^{(1)} = \sigma_j^{(2)}$, then
\begin{align}\label{eq:thm1pf2}
\Theta_{ij}^{*(1)} = \Theta_{ij}^{*(2)} \quad\textrm{and}\quad \Theta_{jj}^{*(1)} = \Theta_{jj}^{*(2)}.
\end{align}

By choosing $M := S \cup \{i, j\}$ and denoting $M^\ast := [j] \setminus M$, recall that the entries of $\Theta_M^{(k)}$ can be written as
\begin{align*}
(\Theta^{(k)}_M)_{i_M j_M} = \Theta^{*(k)}_{ij} - \Theta^{*(k)}_{i, M^\ast} (\Theta^{*(k)}_{M^\ast, M^\ast})^{-1} \Theta^{*(k)}_{M^\ast, j}.
\end{align*}
Now by invoking Lemma \ref{lem:skel1} and Eq.~\eqref{eq:thm1pf2}, we obtain that  $(\Theta^{(1)}_M)_{i_M j_M} = (\Theta^{(2)}_M)_{i_M j_M}$ and $(\Theta^{(1)}_M)_{j_M j_M} = (\Theta^{(2)}_M)_{j_M j_M}$. Finally, using Eq.~\eqref{eq:inv}, we obtain $\beta_{i,j\mid S}^{(1)} = \beta_{i,j\mid S}^{(2)}$.

\textbf{Case 2} $\sigma_j^{(1)} \neq \sigma_j^{(2)}$

In this case, we prove that regressing on all of the parents of $j$ in both DAGs, i.e., choosing the conditioning set as $S = \pa^{(1)}(j) \cup \pa^{(2)}(j) \setminus \{i\}$, implies regression invariance. This is a valid choice for $S$, i.e. $S \subseteq S_\Theta \setminus \{i, j\}$, since Assumption ~\ref{ass:prec} ensures that if $\sigma_j^{(1)} \neq \sigma_j^{(2)}$ then $\ell \in S_\Theta$ for all $\ell \in \pa^{(k)}(j)$.

Let $M := S \cup \{i\}$. By regressing $\Xk_j$ onto $\Xk_M$, we get the regression coefficient as
$$
\Xk_j = (\beta_M^{(k)})^T \Xk_M + \tilde{\epsilon}_j^{(k)}.
$$

Let $(\beta_M^{(k)})_{\ell_M}$ denote the $\ell_M$-th entry of $\beta_M^{(k)}$. By the Markov property, when regressing $\Xk_j$ onto $\Xk_M$ where $\pa^{(k)}(j) \subseteq M \subseteq [j-1]$, it is guaranteed that $(\beta_{M}^{(k)})_{\ell_M} = \Bk_{\ell j}$ if $\ell \in \pa^{(k)}(j)$ and $(\beta_{M}^{(k)})_{\ell_M} = 0$ otherwise. Therefore, we have that $\beta_{i,j \mid S}^{(k)} = (\beta_{M}^{(k)})_{i_M} = \Bk_{ij}$, which completes the proof.
\end{proof}

We now show how the proof of Theorem~\ref{thm:skel} follows from this lemma.
\begin{proof}
By applying Assumption~\ref{ass:prec} we have that $\bar{\Delta} \subseteq \Delta_\Theta$. Then the proof of Theorem~\ref{thm:skel} follows trivially from Lemma~\ref{lem:skel2}, since Lemma~\ref{lem:skel2} shows that an edge $i - j$ is deleted during testing the invariance of regression coefficients if and only if $i - j \not\in \bar{\Delta}$.
\end{proof}

We end this section with a remark about the strength of Assumption~\ref{ass:prec}.

\begin{remark}[Strength of Assumption~\ref{ass:prec}] \label{rk:prec}
To analyze the strength of Assumption~\ref{ass:prec}, consider instead the following stronger assumption:

\vspace{0.05cm}
\noindent {\bf Assumption~D.1'} {\it For any choices of $i, j \in [p]$, it holds that
\begin{enumerate}
\item If $\B1_{ij} \neq \B2_{ij}$, then $\Theta_{ij}^{(1)} \neq \Theta_{ij}^{(2)}$, and $\Theta_{i\ell}^{(1)} \neq \Theta_{i\ell}^{(2)}$ for any $\ell$ with directed path $i \rightarrow j \leftarrow \ell$ in either $\G^{(1)}$ or $\G^{(2)}$.
\item If $\sigma_j^{(1)} \neq \sigma_j^{(2)}$, then $\Theta_{jj}^{(1)} \neq \Theta_{jj}^{(2)}$, and $\Theta_{ii}^{(1)} \neq \Theta_{ii}^{(2)} ~\forall~ i \in \pa^{(1)}(j) \cup \pa^{(2)}(j)$.
\end{enumerate}}
\vspace{0.05cm}

Assumption~D.1' is a strictly stronger assumption than Assumption~\ref{ass:prec}, i.e., Assumption~\ref{ass:prec} is satisfied whenever Assumption~D.1' is satisfied. We expect Assumption~D.1' to be much weaker than Assumptions~\ref{ass:adj_faith} and~\ref{ass:ori_faith} in the finite sample regime, and therefore the same also holds for Assumption~\ref{ass:prec}. This is because the number of hypersurfaces violating Assumption~D.1' scales at most as $\mathcal{O}(p^4)$, which is a much smaller number as compared to Assumptions~\ref{ass:adj_faith} and~\ref{ass:ori_faith} that scale as $\mathcal{O}(|\Delta_\Theta| 2^{| S_\Theta | - 1})$.
\qed
\end{remark}

\subsection{Proof of Theorem~\ref{thm:dir}} \label{sec:pfdir}

In this section, we provide a proof of Theorem~\ref{thm:dir} when Algorithm~\ref{alg:dir} is initialized in the difference-UG. 

\begin{lemma} \label{lem:dir1}
For all nodes $j$ incident to at least one edge in $\bar{\Delta}$, 
$\sigma_j^{(1)} = \sigma_j^{(2)}$ if and only if \;
		$\exists~ S  \subseteq S_\Theta \setminus \{i, j\} \st \sigma_{j \mid S}^{(1)} = \sigma_{j \mid S}^{(2)}$.
\end{lemma}

\begin{proof}
Proving the ``if'' direction is equivalent to showing that, if $\sigma_j^{(1)} \neq \sigma_j^{(2)}$, then 
\begin{align}\label{eq:thm2pf1}
\forall S \subseteq S_\Theta \setminus \{j\},\; \sigma_{j \mid S}^{(1)} \neq \sigma_{j \mid S}^{(2)}.
\end{align}
This follows directly from Assumption~\ref{ass:ori_faith}.

To prove the ``only if'' direction, consider again the marginal distribution of $\Xk_{1:j}$. Since $\sigma_j^{(1)} = \sigma_j^{(2)}$, we have that $\Theta_{jj}^{*(1)} = \Theta_{jj}^{*(2)}$. Let $M := S_{\Theta^*} \cup \{j\}$ and let $S := M \setminus \{j\}$, since $(\sigma_{j \mid S}^{(k)})^2 = ((\Theta^{(k)}_M)_{j_M j_M})^{-1}$ and $(\Theta^{(1)}_M)_{j_M j_M} = (\Theta^{(2)}_M)_{j_M j_M}$ by using Lemma~\ref{lem:skel1}, we have that $\sigma_{j \mid S}^{(1)} = \sigma_{j \mid S}^{(2)}$.
\end{proof}


\begin{lemma} \label{lem:dir2}
	$\forall~ i - j \in \bar{\Delta}$ such that $\sigma_j^{(1)} = \sigma_j^{(2)}$ it holds that, 
	\begin{enumerate}
		\item if $i \rightarrow j \in \Delta$, then $i \in S$ for all $ S \st \sigma_{j \mid S}^{(1)} = \sigma_{j \mid S}^{(2)}$.
		\item if $j \rightarrow i \in \Delta$, then $i \not\in S$ for all $S \st \sigma_{j \mid S}^{(1)} = \sigma_{j \mid S}^{(2)}$.
	\end{enumerate}
\end{lemma}

\begin{proof}
We prove both statements by contradiction. For $\B1_{ij} \neq \B2_{ij}$, suppose there exists a $S$ such that $\sigma_{j \mid S}^{(1)} = \sigma_{j \mid S}^{(2)}$ while $i \not\in S$. This contradicts Assumption~\ref{ass:ori_faith}.

Similarly, in the second statement for $\B1_{ji} \neq \B2_{ji}$, suppose there exists $S$ such that $\sigma_{j \mid S \cup \{i\}}^{(1)} = \sigma_{j \mid S \cup \{i\}}^{(2)}$. This also contradicts Assumption~\ref{ass:ori_faith}.
\end{proof}

We now show how the proof of Theorem~\ref{thm:dir} follows from this lemma.


\begin{proof}
By Lemma~\ref{lem:dir1}, there exists $S$ such that 
$\sigma_{j \mid S}^{(1)} = \sigma_{j \mid S}^{(2)}$ 
if and only if $\sigma_j^{(1)} = \sigma_j^{(2)}$. Therefore, all the nodes where the internal noise variance is unchanged will be chosen by Algorithm~\ref{alg:dir}. In addition, it also follows from  Lemma~\ref{lem:dir2} that for any $i \to j \in \Delta$, $i \in S$ and for any $j \to i \in \Delta$, $i \not \in S$. This proves that for any node $i$ where $\sigma_i^{(k)}$ is invariant, all edges adjacent to $i$ are oriented and that all edges oriented before the last step of Algorithm~\ref{alg:dir} are correctly oriented. 

It remains to show that all edges oriented in the last step of Algorithm~\ref{alg:dir} are correct. This easily follows from the acyclicity property of the underlying graphs and from the fact that all edge orientations before the last step are correct.
\end{proof}

\section{Examples for Remark~\ref{rk:assumption}} \label{sec:justification}

Since our assumptions are closely related to the faithfulness assumption, it is interesting to compare the entailment relationship between our assumptions, i.e., Assumptions~\ref{ass:adj_faith} and~\ref{ass:ori_faith}, and the faithfulness assumption. In this section, we give the following two counterexamples to show that our assumptions and the faithfulness assumption do not imply one another. 

\begin{example}
	\label{ex:comp1}
	We give a 3-node example that satisfies Assumptions~\ref{ass:adj_faith} and~\ref{ass:ori_faith} but does not satisfy the faithfulness assumption. Consider two linear SEMs $(\B1, \epsilon^{(1)})$ and $(\B2, \epsilon^{(2)})$ with $\epsilon^{(k)}_j \sim \N(0, 1) ~\forall~ j,k$ and where $\B1$ and $\B2$ are the autoregressive matrices defined as shown in Figure~\ref{fig:comp1}. Clearly, $\PP^{(1)}$ does not satisfy the faithfulness assumption with respect to $\G^{(1)}$ since nodes $1$ and $3$ are d-connected given $\emptyset$, but $X_1^{(1)} \independent X_3^{(1)}$. However, it is a short exercise to show that for all choices of $S$, i.e. $\emptyset$ and $\{ 2 \}$, we have $\beta_{1,3 \mid S}^{(1)} \neq \beta_{1,3 \mid S}^{(2)}$, $\beta_{3,1 \mid S}^{(1)} \neq \beta_{3,1 \mid S}^{(2)}$, $\sigma_{3 \mid S}^{(1)} \neq \sigma_{3 \mid S}^{(2)}$ and $\sigma_{1 \mid S \cup \{3\}}^{(1)} \neq \sigma_{1 \mid S \cup \{3\}}^{(2)}$. Therefore, this example satisfies Assumptions~\ref{ass:adj_faith} and~\ref{ass:ori_faith}.
\end{example}

\begin{figure}[t]
\centering
{\includegraphics[width=.25\textwidth]{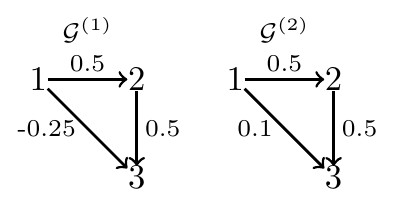}}
\caption{Example of two linear SEMs that satisfy Assumpitons~\ref{ass:adj_faith} and~\ref{ass:ori_faith} but do not satisfy the faithfulness assumption. The autoregressive matrices $\B1$ and $\B2$ are shown as edge weights in $\G^{(1)}$ and $\G^{(2)}$. We assume that all noise terms are standard normal random variables.}  \label{fig:comp1}
\end{figure}

\begin{figure}[b]
	\centering
	{\includegraphics[width=.25\textwidth]{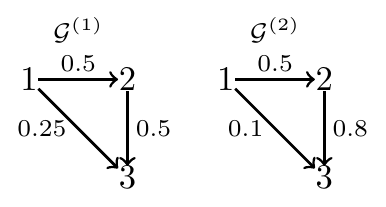}}
	\caption{Example of two linear SEMs that satisfy the faithfulness assumption but do not satisfy Assumption~\ref{ass:adj_faith}. The autoregressive matrices $\B1$ and $\B2$ are shown as edge weights in $\G^{(1)}$ and $\G^{(2)}$. We assume that all noise terms are standard normal random variables.}  \label{fig:comp2}
\end{figure}

\begin{example}
	\label{ex:comp2}
	We give a 3-node example that satisfies the faithfulness assumption but does not satisfy Assumption~\ref{ass:adj_faith}. Consider two linear SEMs where all $\epsilon^{(k)}_j$ are standard normal random variables and $\B1$ and $\B2$ are defined as shown in Figure~\ref{fig:comp2}. Although $\B1_{13} \neq \B2_{13}$, by choosing $S = \emptyset$, we still have that $\beta_{1,3 \mid S}^{(1)} = \beta_{1,3 \mid S}^{(2)} = 0.5$. Therefore, although both SEMs satisfy the faithfulness assumption, the pair does not satisfy Assumption~\ref{ass:adj_faith}.
\end{example}

Next, we give an example explaining the hypersurfaces that correspond to the set of parameters violating our assumptions versus the faithfulness assumption. This example shows that the number of hypersurfaces corresponding to violations of the faithfulness assumption is much higher than the number of hypersurfaces corresponding to violations of our assumptions, which implies that the faithfulness assumption is more restrictive in the finite sample regime.

\begin{figure}[t]
	\vspace{-0.3cm}
\centering
{\includegraphics[width=.25\textwidth]{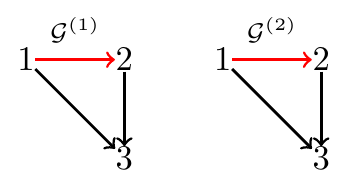}}
\vspace{-0.3cm}
\caption{Example of two fully connected linear SEMs. The red edges correspond to the edges that have different edge weights across the two DAGs, the black edges correspond to the edges that have the same edge weights across the two DAGs. The variances of internal noise terms remain the same for both DAGs.}
\vspace{-0.3cm}
\label{fig:comp3}
\end{figure}

\begin{example}
\label{ex:comp3}
We give a 3-node example to provide intuition for why the number of hypersurfaces violating the faithfulness assumption is usually much higher than the number of hypersurfaces violating our assumptions. Consider the two fully connected linear SEMs $(\B1, \epsilon^{(1)})$ and $(\B2, \epsilon^{(2)})$ shown in Figure~\ref{fig:comp3}. In this example, $\B1_{12} \neq \B2_{12}$ while the noise variances and all other edge weights are not changed across the two DAGs.

\begin{figure}[b!]
	\centering
	{\includegraphics[width=.9\textwidth]{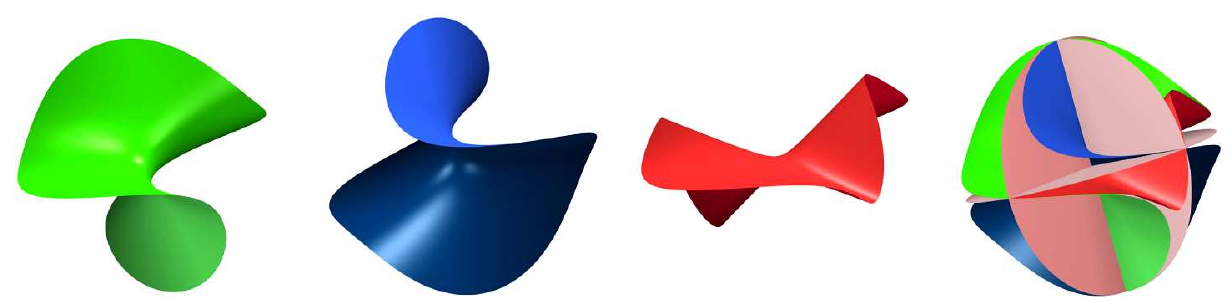}}
	\vspace{-0.2cm}
	\caption{Parameter values corresponding to unfaithful distributions in Example~\ref{ex:comp3}; the first three figures are the hypersurfaces corresponding to $\cov(X_1, X_2) = 0$, $\cov(X_1, X_2 \mid X_3) = 0$ and $\cov(X_1, X_3) = 0$ respectively when setting $\sigma_i=1$ for visualization in 3d; the last figure shows the hypersurfaces of the first $6$ polynomials with $\sigma_i=1$. Figure adopted from \cite[Figure~2]{URB13})}  \label{fig:hypersurfaces}
	\vspace{-0.2cm}
\end{figure}

If we think of each parameter $\Bk_{ij}$ or $\sigma_j^{(k)}$ not as a parameter but rather as an indeterminate, the set of parameters that violate the faithfulness assumption and our assumptions correspond to a system of polynomial equations in the following $7$ indeterminates: $(\B1_{12}, \B2_{12}, B_{13}, B_{23}, \sigma_1, \sigma_2, \sigma_3)$. Note that here we use a single indeterminate $B_{13}$ to encode both the parameters $\B1_{13}$ and $\B2_{13}$ since they have the same value. The set of parameters that violate the faithfulness assumption are given by the following $11$ polynomial equations and hence correspond to a collection of 11 hypersurfaces:
\begin{align*}
\cov(X^{(1)}_1, X^{(1)}_2) :\;&\;\; \B1_{12} \sigma_1^2 = 0,\\
\cov(X^{(1)}_1, X^{(1)}_3) :\;&\;\; B_{13} \sigma_1^2 + \B1_{12} B_{23} \sigma_1^2 = 0, \\
\cov(X^{(1)}_2, X^{(1)}_3) :\;&\;\; (\B1_{12})^2 B_{23} \sigma_1^2 + \B1_{12}B_{13} \sigma_1^2 + B_{23} \sigma_2^2 = 0, \\
\cov(X^{(1)}_1, X^{(1)}_2 \mid X^{(1)}_3) :\;&\;\;  - \frac{B_{13}B_{23}\sigma_1^2\sigma_2^2 - \B1_{12}\sigma_1^2\sigma_3^2}{(B_{13}  + \B1_{12}B_{23})^2 \sigma_1^2 + B_{23}^2 \sigma_2^2 + \sigma_3^2} = 0, \\
\cov(X^{(1)}_1, X^{(1)}_3 \mid X^{(1)}_2) :\;&\;\; \frac{B_{13}\sigma_1^2\sigma_2^2}{(\B1_{12})^2 \sigma_1^2 + \sigma_2^2} = 0, \\
\cov(X^{(1)}_2, X^{(1)}_3 \mid X^{(1)}_1) :\;&\;\; B_{23} \sigma_2^2 = 0, \\
\cov(X^{(2)}_1, X^{(2)}_2) :\;&\;\; \B2_{12} \sigma_1^2 = 0,\\
\cov(X^{(2)}_1, X^{(2)}_3) :\;&\;\; B_{13} \sigma_1^2 + \B2_{12} B_{23} \sigma_1^2 = 0, \\
\cov(X^{(2)}_2, X^{(2)}_3) :\;&\;\; (\B2_{12})^2 B_{23} \sigma_1^2 + \B2_{12}B_{13} \sigma_1^2 + B_{23} \sigma_2^2 = 0, \\
\cov(X^{(2)}_1, X^{(2)}_2 \mid X^{(2)}_3) :\;&\;\;  - \frac{B_{13}B_{23}\sigma_1^2\sigma_2^2 - \B2_{12}\sigma_1^2\sigma_3^2}{(B_{13}  + \B2_{12}B_{23})^2 \sigma_1^2 + B_{23}^2 \sigma_2^2 + \sigma_3^2} = 0, \\
\cov(X^{(2)}_1, X^{(2)}_3 \mid X^{(2)}_2) :\;&\;\; \frac{B_{13}\sigma_1^2\sigma_2^2}{(\B2_{12})^2 \sigma_1^2 + \sigma_2^2} = 0.
\end{align*}

To get a better sense of how the hypersurfaces of these polynomials are distributed in the parameter space, Figure~\ref{fig:hypersurfaces} visualizes the first $6$ hypersurfaces. This figure was directly adopted from Figure~2 of~\cite{URB13}. On the other hand, the polynomials of the parameters violating our assumptions are as follows:
\begin{align*}
\beta_{1,2 \mid \emptyset}^{(1)} - \beta_{1,2 \mid \emptyset}^{(2)} :\;&\;\; \B1_{12} - \B2_{12} = 0,\\
\beta_{2,1 \mid \emptyset}^{(1)} - \beta_{2,1 \mid \emptyset}^{(2)} :\;&\;\; \frac{\B1_{12} \sigma_1^2}{(\B1_{12})^2 \sigma_1^2 + \sigma_2^2} - \frac{\B2_{12} \sigma_1^2}{(\B2_{12})^2 \sigma_1^2 + \sigma_2^2} = 0, \\
(\sigma_{2 \mid \emptyset}^{(1)})^2 - (\sigma_{2 \mid \emptyset}^{(2)})^2 :\;&\;\; (\B1_{12})^2 \sigma_1^2 - (\B2_{12})^2 \sigma_1^2 = 0, \\
(\sigma_{1 \mid \{2\}}^{(1)})^2 - (\sigma_{1 \mid \{2\}}^{(2)})^2 :\;&\;\; \frac{1}{(\B1_{12})^2 \sigma_2^{-2} + \sigma_1^{-2}} - \frac{1}{(\B2_{12})^2 \sigma_2^{-2} + \sigma_1^{-2}} = 0.
\end{align*}
Clearly, the number of polynomials that violate Assumptions~\ref{ass:adj_faith} and~\ref{ass:ori_faith} is much smaller as compared to those of the faithfulness assumption. HAs a consequence our assumption is weaker than the faithfulness assumption in the finite sample regime (where violations correspond to points that are close to any of the hypersurfaces).
\end{example}

\section{Constraint-based method for estimating the difference-UG}

In this section, we present a constraint-based method for estimating the difference-UG model in linear SEMs with general additive noise, i.e., where the noise is not necessarily Gaussian. Our constraint-based method is built on performing a hypothesis test on each $(i,j)$-th entry and then finding the set of $(i,j)$-th entries where $\Theta_{ij}^{(1)} \neq \Theta_{ij}^{(2)}$. The test for invariance of diagonal entries, i.e., $\Theta_{ii}^{(k)}$, is equivalent to the hypothesis test $H_0^{i \mid [p] \setminus \{i\}}$ as discussed in Section~\ref{sec:method}, since $(\sigma_{i \mid [p] \setminus \{i\}}^{(k)})^2 = (\Theta_{ii}^{(k)})^{-1}$. For the non-diagonal entries, since the non-zero pattern of $\Theta_{ij}^{(k)}$ is the same as the non-zero pattern of the partial correlation coefficients, i.e., $\rho_{ij \mid [p] \setminus \{i,j\}}^{(k)}$, we first find the set of non-diagonal entries that are different between $\Theta^{(1)}$ and $\Theta^{(2)}$ by doing partial correlation tests for each distribution and then comparing the non-zero patterns. After that, for each entry $(i,j)$ that is estimated to be non-zero in both $\Theta^{(1)}$ and $\Theta^{(2)}$, we use the test statistic:
\begin{align*}
\begin{split}
&\hat{Q} := \left(\hat{\Theta}_{ij}^{(1)} - \hat{\Theta}_{ij}^{(2)}\right)^2 \cdot \\
&\qquad \left(\!\frac{\hat{\Theta}_{ii}^{(1)}\hat{\Theta}_{jj}^{(1)} \!\!+\! (\hat{\Theta}_{ij}^{(1)})^2}{n_1}  \!+\! \frac{\hat{\Theta}_{ii}^{(2)}\hat{\Theta}_{jj}^{(2)} \!\!+\! (\hat{\Theta}_{ij}^{(2)})^2}{n_2}\!\right)^{-1}
\end{split}
\end{align*}
and test if it fits the F-distribution with parameters $F(1, n_1 + n_2 - 2p + 2)$. If this is the case, we conclude that this particular entry $(i,j)$ is invariant between the two precision matrices. The consistency guarantees of $H_0^{i \mid [p] \setminus \{i\}}$ and partial correlation tests follow trivially from previous results. For $\hat{Q}$, it follows from Proposition~3 of~\cite{DP07} on the asymptotic normal distribution of the empirical precision matrix $\hat{\Theta}$ that if the null hypothesis is true, then $\hat{Q}$ converges in distribution to $\chi^2(1)$ as $n_1,n_2 \to \infty$. 

\section{Additional high-dimensional evaluation}

\textbf{High-dimensional setting: 10\% changes.}
We present the results of increasing the number of changes between the two DAGs, and hence the size of $S_\Theta$. We used the same simulation parameters as for Figure~\ref{fig:high-dim}, i.e. $p=100$ nodes, a neighbourhood size of $s=10$, and sample size $n=300$, except that the total number of changes was $10\%$ of the number of edges in $B^{(1)}$, rather than 5\%. As shown in Figure~\ref{fig:high-dim-supp}, both initializations of the DCI algorithm still outperform separate estimation by GES and the PC algorithm. However, because the underlying DAGs have maintained constant sparsity while the difference-DAG has become more dense, the gains in performance by using the DCI algorithm have slightly diminished, as expected by our theoretical analysis.

\begin{figure}[t]
	\centering
	\subfigure[difference-DAG skeleton $\bar{\Delta}$]
	{\includegraphics[width=.4\textwidth]{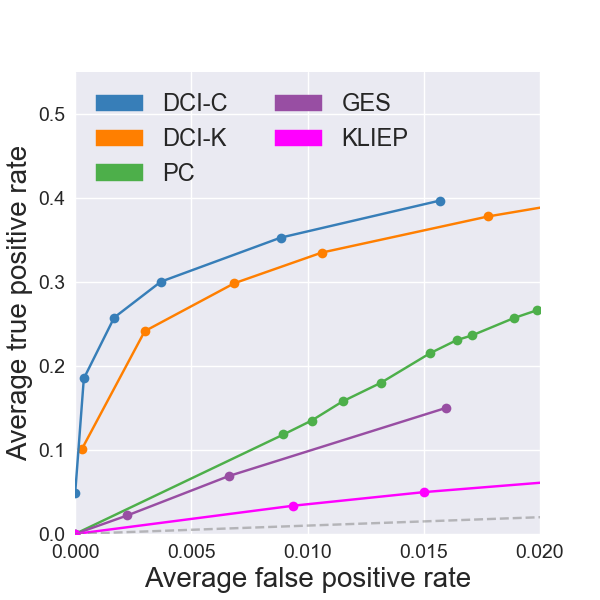}}
	\qquad\qquad\subfigure[difference-DAG $\Delta$]
	{\includegraphics[width=.4\textwidth]{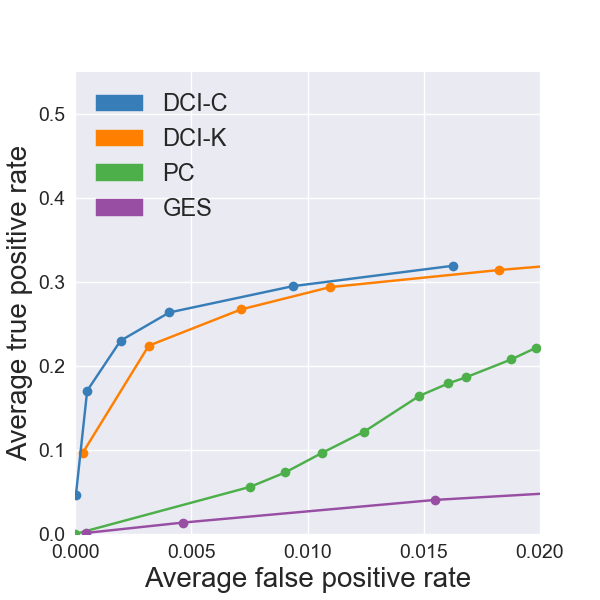}}
	\vspace{-0.2cm}
	\caption{ROC curves for estimating the difference-DAG $\Delta$ and its skeleton $\bar{\Delta}$ with $p = 100$ nodes, expected neighbourhood size $s = 10$, $n = 300$ samples, and 10\% percent change~between~DAGs.}
	\label{fig:high-dim-supp}
\end{figure}

\section{Real data analysis - ovarian cancer}

We tested our method on an ovarian cancer data set~\cite{TTG08}. This data set consists of the gene expression data of patients with ovarian cancer. The patients are divided into six subtypes (C1-C6). The C1 subtype was characterized by differential expression of genes associated with stromal and immune cell types and is associated with shorter survival rates. In this experiment, we divide the subjects into two groups, group 1 with $n_1 = 78$ subjects containing patients with C1 subtype, and group 2 with $n_2 = 113$ subjects containing patients with C2-C6 subtypes. In this work, we focused on two pathways from the KEGG database~\cite{KGS11,OGS99}, the apoptosis pathway containing $87$ genes, and the TGF-$\beta$ pathway with $82$ genes.

\begin{figure}[b!]
	\centering
	\subfigure[Apoptosis pathway]{\includegraphics[width=.35\textwidth]{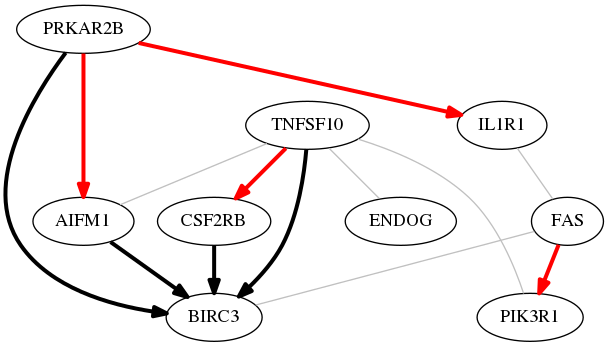}}
	\qquad\qquad\subfigure[TGF-$\beta$ pathway]{\includegraphics[width=.35\textwidth]{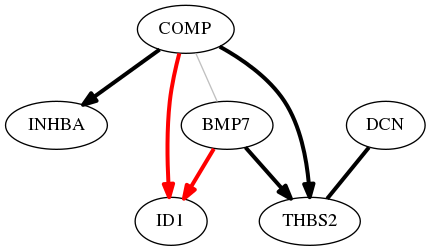}}
	\vspace{-0.2cm}
	\caption{Estimate of the difference DAG between the two groups for the apoptosis and TGF-$\beta$ pathways. The black lines represent the edges discovered by both our method and DPM, the red lines represent the edges discovered only by our method, and the grey lines represent the undirected edges discovered only by DPM.}
	\label{fig:ovarian}
	\vspace{-0.2cm}
\end{figure}

We compared our results to those obtained by the DPM method~\cite{ZCL14}, which infers the difference in the undirected setting. As input to Algorithm \ref{alg:skel}, we took $S_\Theta$ to be all of the nodes in the output of the DPM algorithm and took $\Delta_\Theta$ to be the fully connected graph on $S_\Theta$. We then learned the difference DAG using Algorithm \ref{alg:dir}. The final set of edges over different tuning parameters was chosen using stability selection as proposed in~\cite{Meinshausen2010} and is shown in Figure~\ref{fig:ovarian}. This procedure identified two hub nodes in the apoptosis pathway: BIRC3 and PRKAR2B. BIRC3 has been shown to be an inhibitor of apoptosis~\cite{Johnstone2008} and is one of the top disregulated genes in ovarian cancer~\cite{Jonsson2014}. This gene has also been recovered by the DPM method as one of the hub nodes. While BIRC3 has high in-degree, hub gene PRKAR2B has high out-degree, making it a better candidate for possible interventions in ovarian cancer since knocking out a gene with high out-degree will have widespread downstream effects on the target genes. Indeed, PRKAR2B is a known important regulatory unit for cancer cell growth~\cite{CBG08} and the RII-$\beta$ protein encoded by PRKAR2B has already been studied as a therapeutic target for cancer therapy~\cite{MGM06,CCY99}. In addition, PRKAR2B has also been shown to play an important role in disease progression in ovarian cancer cells~\cite{CNW08}. Since the DPM method does not infer directionality, it is not possible to tell which of the hub genes might be a better interventional target. This is remedied by our method and its impact for identifying possible therapeutic targets in real data is showcased by finding an already known drug target for cancer.

\begin{figure}[t!]
	\centering
	\subfigure[apoptosis, PC]
	{\includegraphics[width=.2\textwidth]{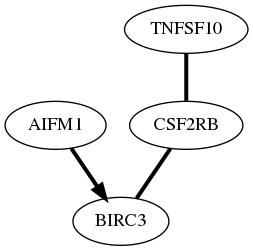}}
	\hspace{.1\textwidth}
	\subfigure[\!\!TGF-$\beta$,~GES]{\includegraphics[scale=0.55]{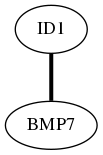}}
	\hspace{.1\textwidth}
	\subfigure[TGF-$\beta$, PC]{\includegraphics[scale = 0.55]{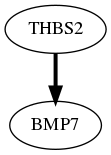}}
	\vspace{-0.2cm}
	\caption{Estimate of the difference DAG between the two groups for the apoptosis and TGF-$\beta$ pathways using the PC and GES algorithms.}
	\label{fig:ovarian_pc_ges}
	\vspace{-0.2cm}
\end{figure}

\begin{figure}[b!]
	\centering
	\subfigure[GES]{\includegraphics[width=.9\textwidth]{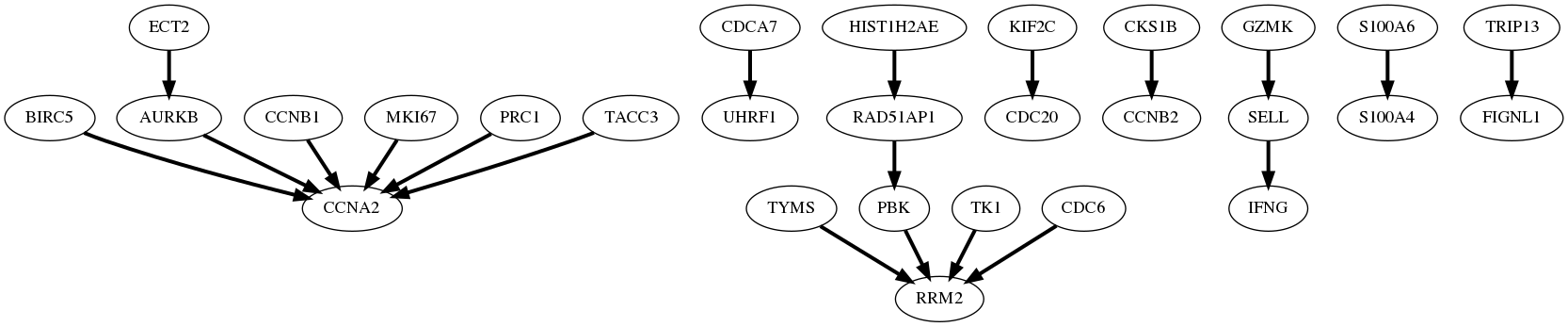}} \\
	\subfigure[PC]{\includegraphics[width=.9\textwidth]{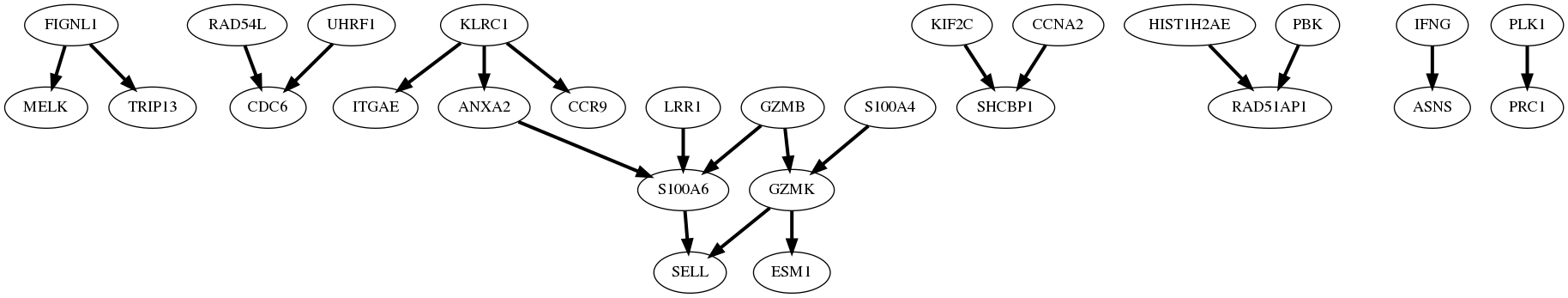}}
	\caption{Estimate of the difference DAG between naive and activated T cells using the PC and GES algorithms.}
	\label{fig:ges_tcell}
\end{figure}

For the TGF-$\beta$ pathway, our analysis identified THBS2 and COMP as hub nodes. Both of these genes have been implicated in resistance to chemotherapy in epithelial ovarian cancer~\cite{MFC13}, confirming the importance of our findings. These nodes were also recovered by DPM. 

Overall, the undirected graph discovered by DPM is similar to the DAG found by our method. The disparity in the TGF-$\beta$ pathway between the difference UG model $\Delta_\Theta$ and the difference DAG model $\Delta$ can be explained by the fact that the edge between COMP$-$BMP7 in $\Delta_\Theta$ can be accounted for by the two edges BMP7$\rightarrow$ID1 and COMP$\rightarrow$ID1 in $\Delta$. Though these edges might represent the true regulatory pathways, the sparsity-inducing penalty in the DPM algorithm could remove them while leaving the edge between COMP and BMP7. This disparity between the two algorithms highlights the importance of replacing correlative reasoning with causal reasoning, and accentuates the significance of our contribution.

We also applied the GES and PC algorithms on the ovarian cancer data set. We considered the set of edges that appeared in one estimated skeleton but disappeared in the other as the estimated skeleton of the D-DAG $\bar{\Delta}$. In determining orientations, we considered the arrows that were directed in one estimated CP-DAG but disappeared in the other as the estimated set of directed arrows. Figure~\ref{fig:ovarian_pc_ges} shows the results by applying the PC algorithm on the apoptosis and TGF-$\beta$ pathway and the results by applying GES on the TGF-$\beta$ pathway. Here we omitted GES results on the apoptosis pathway since GES algorithm did not discover any differences on the apoptosis pathway. Figure~\ref{fig:ovarian_pc_ges} shows that PC and GES cannot discover any hub nodes.

\section{Real data analysis - T cell activation}

We compare DCI with the GES and PC algorithms on the T cell activation data set. Figure~\ref{fig:ges_tcell} (a) shows the results of applying GES to naive and activated data sets separately and calculating the difference. Figure~\ref{fig:ges_tcell} (b) shows the estimated results of applying PC to the T cell data set. 

\end{document}